\documentclass[a4paper,11pt]{article}

\usepackage{fullpage}
\usepackage{soul}
\usepackage[utf8]{inputenc}
\usepackage[small]{caption}
\usepackage{xcolor}
\usepackage{amsmath}
\usepackage{booktabs}
\usepackage{algorithm}
\usepackage{enumitem}
\usepackage{framed}
\usepackage{tikz}
\usepackage{url}

\usepackage{subcaption,algpseudocode}

\usepackage{amsthm}
\usepackage{amssymb}
\usepackage{afterpage}

\newtheorem{theorem}{Theorem}[section]
\newtheorem{corollary}[theorem]{Corollary}
\newtheorem{proposition}[theorem]{Proposition}
\newtheorem{lemma}[theorem]{Lemma}
\newtheorem{conjecture}[theorem]{Conjecture}

\theoremstyle{definition}
\newtheorem{definition}[theorem]{Definition}

\usepackage{natbib}

\allowdisplaybreaks

\usepackage{authblk}

\title{\bf The Price of Connectivity in Fair Division}

\author[1]{Xiaohui Bei} 
\author[2]{Ayumi Igarashi}
\author[3]{Xinhang Lu}
\author[4]{Warut Suksompong}

\affil[1]{Nanyang Technological University, Singapore}
\affil[2]{National Institute of Informatics, Japan}
\affil[3]{University of New South Wales, Australia}
\affil[4]{National University of Singapore, Singapore}

\date{}

\begin{document}

\maketitle

\begin{abstract}
We study the allocation of indivisible goods that form an undirected graph and quantify the loss of fairness when we impose a constraint that each agent must receive a connected subgraph.
Our focus is on well-studied fairness notions including envy-freeness and maximin share fairness.
We introduce the \emph{price of connectivity} to capture the largest multiplicative gap between the graph-specific and the unconstrained maximin share, and derive bounds on this quantity which are tight for large classes of graphs in the case of two agents and for paths and stars in the general case.
For instance, with two agents we show that for biconnected graphs it is possible to obtain at least $3/4$ of the maximin share with connected allocations, while for the remaining graphs the guarantee is at most $1/2$.
In addition, we determine the optimal relaxation of envy-freeness that can be obtained with each graph for two agents, and characterize the set of trees and complete bipartite graphs that always admit an allocation satisfying envy-freeness up to one good (EF1) for three agents.
Our work demonstrates several applications of graph-theoretic tools and concepts to fair division problems.
\end{abstract}

\section{Introduction}

We consider a classical resource allocation setting where a set of goods is to be allocated among interested agents. Our goal is to find an allocation that is fair to all agents.
This problem has been addressed in a large body of literature on \emph{fair division} \citep{BramsTa96,Moulin03}, which has found applications ranging from divorce settlement \citep{BramsTa96-2} to credit assignment \citep{DeclippelMoTi08}.
The two most prominent fairness notions in the literature are \emph{envy-freeness} and \emph{proportionality}.
An allocation is said to be envy-free if every agent likes her bundle at least as much as any other agent's bundle, and proportional if every agent receives value at least $1/n$ of her value for the entire set of goods, where $n$ denotes the number of agents.

Our focus in this paper is on the setting where we allocate \emph{indivisible} goods.
This pertains to the allocation of houses, cars, artworks, electronics, and many other common items. When goods are indivisible, neither envy-freeness nor proportionality can always be fulfilled, e.g., when two agents try to divide a single valuable good.
As a result, relaxations of both notions have been studied. Envy-freeness is often relaxed to \emph{envy-freeness up to one good (EF1)}---this means  that any envy that an agent has towards another agent can be eliminated by removing a good from the latter agent's bundle.
An EF1 allocation exists for any number of agents with arbitrary monotonic utilities \citep{LiptonMaMo04}.
Likewise, proportionality can be relaxed to \emph{maximin share fairness}---the \emph{maximin share (MMS)} of an agent is the largest value that the agent can guarantee for herself if she is allowed to divide the goods into $n$ parts and always receives the worst part.
An allocation that gives every agent her maximin share---said to satisfy maximin share fairness---does not always exist for additive utilities, but a constant multiplicative approximation can be obtained \citep{KurokawaPrWa18}.

Perhaps the most well-known fair division protocol is the \emph{cut-and-choose protocol}, which dates back to at least the Bible and can be used to allocate a \emph{divisible} good between two agents.
In this protocol, the first agent divides the good into two equal parts (this is possible because the good is divisible), and the second agent chooses the part that she prefers.
The cut-and-choose protocol has a direct analogue in the indivisible goods setting: since an equal partition may no longer exist, the first agent now divides the goods into two parts that are as equal as possible in her view.
The resulting allocation is guaranteed to satisfy both maximin share fairness and EF1.\footnote{In fact, it also satisfies a notion called \emph{envy-freeness up to any good (EFX)}, which is stronger than EF1 \citep{PlautRo20}.}
However, these guarantees rely crucially on the assumption that any allocation of the goods to the two agents can be chosen---in reality, there are often constraints on the allocations that we desire.
One common type of constraints is captured by a model of \citet{BouveretCeEl17}, where the goods are vertices of a connected undirected graph and each agent must be allocated a connected subgraph.
For instance, the goods could represent offices in a university building that we wish to divide between research groups, and it is desirable for each group to receive a connected set of offices in order to facilitate communication within the group.
To what extent do the fairness guarantees continue to hold when connectivity constraints are imposed, and how does the answer depend on the underlying graph?
Put another way, what is the price in terms of fairness that we have to pay if we desire connectivity?

\subsection{Our Contributions}

In this paper, we make several contributions to the active line of work on fairly allocating indivisible goods under connectivity constraints.\footnote{We survey this line of work in Section~\ref{sec:relatedwork}.}
While we also provide fairness guarantees for any number of agents, the majority of our results concern the setting of two agents.
We emphasize here that this setting is fundamental in fair division.
Indeed, a number of fair division applications including divorce settlements, inheritance division, and international border disputes often fall into this setting, and numerous prominent works in the field deal exclusively with the two-agent case (e.g., \citep{BramsFi00,BramsKiKl12,BramsKiKl14,KilgourVe18}).\footnote{See also \cite[Section~1.1.1]{PlautRo20-2} for further discussion on the importance of the two-agent setting.}
In addition, as we will see, under connectivity constraints the setting with two agents is already surprisingly rich and gives rise to several mathematically deep and challenging questions.

\renewcommand{\arraystretch}{1.2}
\begin{table*}
\centering
    \begin{tabular}{| c | c | c |}
    \hline
    Class of graphs & $n=2$ & $n\geq 2$ \\ \hline \hline
     Paths & \begin{tabular}[c]{@{}c@{}}$1$ if $m=2$\\$2$ if $m\geq 3$\end{tabular} & \begin{tabular}[c]{@{}c@{}}$1$ if $m<n$\\$m-n+1$ if $n\leq m< 2n-1$\\ $n$ if $m\geq 2n-1$ \end{tabular}  \\ \hline
     Stars & $m-1$ & $m-n+1$ \\  \hline
     Vertex connectivity $1$ & $k$ (see caption) & $\leq m-n+1$  \\ \hline
     Vertex connectivity $2$ & $\frac{4}{3}$  & $\leq m-n+1$  \\  \hline
     Vertex connectivity $\geq 3$ & $\leq \frac{4}{3}$  & $\leq m-n+1$  \\
    \hline
    \end{tabular}
    \vspace{5mm}
    \caption{Summary of our PoC bounds, where $n$ and $m$ denote the number of agents and goods, respectively. For $n=2$ and graphs with vertex connectivity $1$ (which include all trees), the parameter $k$ denotes the maximum number of components that can result from deleting a single vertex from the graph.}
    \label{table:summary}
\end{table*}

We begin by studying maximin share fairness for agents with additive utilities.
We define the \emph{price of connectivity (PoC)} of a graph to be the largest multiplicative gap between the maximin share defined over all possible partitions and the \emph{graph maximin share (G-MMS)}, which is defined over all partitions that respect the connectivity constraints of the graph.
For any graph and any number of agents, it follows from the definitions that if the PoC is $\alpha$ and one can give each agent $\beta$ times her G-MMS, then it is also possible to guarantee all agents a $\beta/\alpha$ fraction of their MMS.
Moreover, in cases where giving every agent their full G-MMS is possible (i.e., $\beta=1$), we observe in \textbf{Section~\ref{sec:prelim}} that the resulting factor $1/\alpha$ is tight---in other words, the PoC is the reciprocal of the optimal MMS approximation that can be achieved.
Since it is known from prior work that $\beta=1$ for two agents and arbitrary graphs as well as for any number of agents and trees, our PoC notion precisely captures the best possible MMS guarantee in these cases.
Hence, determining the PoC, whose definition only involves a single utility function, allows us to identify the optimal MMS guarantee for agents with possibly different utility functions.

With this relationship in hand, we proceed to determine the PoC of various graphs; our results are summarized in Table~\ref{table:summary}.
In the two-agent case (\textbf{Section~\ref{sec:mms-two}}), we show that the PoC is related to the \emph{vertex connectivity} of the graph, i.e., the minimum number of vertices whose deletion disconnects the graph. For graphs with connectivity exactly $1$, including all trees, we show that the PoC is equal to the maximum number of connected components that result from deleting one vertex.
As a consequence, the PoC is at least $2$ for any graph in this class. On the other hand, we show an upper bound of $4/3$ for all graphs with connectivity at least $2$---this bound is tight for all graphs with connectivity exactly $2$ and, perhaps surprisingly, for certain graphs with connectivity up to $5$.
In addition, we pose an intriguing conjecture that the PoC of any graph with connectivity at least $2$ is closely related to its ``linkedness''---the two-agent case would be completely solved if the conjecture holds---and verify our conjecture when the graph is a complete graph with an arbitrary matching removed.

For any number of agents (\textbf{Section~\ref{sec:mms-any}}), we establish a general upper bound of $m-n+1$ on the PoC (where $m$ and $n$ denote the number of goods and agents, respectively), and show that this implies the existence of a connected allocation that gives every agent at least a $1/(m-n+1)$ fraction of her MMS with respect to any graph. We also derive the exact PoC for paths and stars.
Notably, in order to establish the PoC for paths, we introduce a new relaxation of proportionality that we call the \emph{indivisible proportional share (IPS) property}. This notion strengthens a number of relaxations of proportionality in the literature while maintaining guaranteed existence, so we believe that it may be of independent interest as well.

Next, in \textbf{Section~\ref{sec:envyfree}} we turn our attention to envy-freeness relaxations and allow agents to have arbitrary monotonic utilities.
In the case of two agents, \citet{BiloCaFl19} characterized the graphs for which an EF1 allocation always exists as the graphs that admit a ``bipolar ordering'' (defined in Section~\ref{sec:prelim}).
While the characterization yields a strong fairness guarantee for this class of graphs, it does not give any guarantee for the remaining graphs.
We generalize this result by establishing the optimal relaxation of envy-freeness for every graph---specifically, for each graph, we determine the smallest $k$ for which an allocation that is \emph{envy-free up to $k$ goods (EF$k$)} always exists with two agents.
Intuitively, the less connected the graph is, the weaker the fairness guarantee we can make, i.e., the higher the price we have to pay.
As a corollary, an EF$(m-2)$ allocation exists for any connected graph, and the bound $m-2$ is tight for stars.
By contrast, we show that the fairness notion \emph{envy-freeness up to any good (EFX)}, which is stronger than EF1, can only be guaranteed for complete graphs with two agents.
We then address the case of three agents, where we characterize the set of trees and complete bipartite graphs that admit an EF1 allocation for arbitrary utilities.
In particular, our result for complete bipartite graphs answers an open question raised by \citet{BiloCaFl19}.

From a technical point of view, our work makes extensive use of tools and concepts from graph theory, including vertex connectivity, linkedness, ear decompositions, bipolar orderings, and block decompositions.
While bipolar ordering and block decomposition have been used by \citet{BiloCaFl19} in the EF1 characterization that we mentioned, the other concepts have not previously appeared in the fair division literature to the best of our knowledge.
We believe that establishing these connections enriches the growing literature and lays the groundwork for fruitful collaborations between researchers across the two well-established fields.

Finally, we remark that with the exception of Theorem~\ref{thm:mms-two-almost-complete}, all of our guarantees are constructive.
In particular, we exhibit polynomial-time algorithms that produce allocations satisfying the guarantees.

\subsection{Related Work}
\label{sec:relatedwork}
The fair allocation of indivisible goods has received considerable attention from various research communities, especially in the last few years.
We refer to surveys by \citet{Thomson16}, \citet{Markakis17}, and \citet{Moulin19} for an overview of recent developments in the area.

The papers most closely related to ours are the two papers that we mentioned, by \citet{BouveretCeEl17} and \citet{BiloCaFl19}. Bouveret et al.~showed that for any number of agents with additive utilities, there always exists an allocation that gives every agent her maximin share when the graph is a tree, but not necessarily when the graph is a cycle.
It is important to note that their maximin share notion corresponds to our G-MMS notion and is defined based on the graph, with only connected allocations with respect to that graph taken into account in an agent's calculation.
As an example of a consequence, even though a cycle permits strictly more connected allocations than a path, it offers less guarantee in terms of the G-MMS.
Our approach of considering the (complete-graph) MMS allows us to directly compare the guarantees that can be obtained for different graphs.

\citet{BiloCaFl19} investigated the same model with respect to relaxations of envy-freeness.
As we mentioned, they characterized the set of graphs for which EF1 can be guaranteed in the case of two agents with arbitrary monotonic utilities.
Moreover, they showed that an EF1 allocation always exists on a path for $n\leq 4$.
Intriguingly, the existence question for $n\geq 5$ remains open, although they showed that an EF2 allocation can be guaranteed for any $n$.

Besides \citet{BouveretCeEl17} and \citet{BiloCaFl19}, a number of other authors have recently studied fairness under connectivity constraints.
\citet{LoncTr18} investigated maximin share fairness in the case of cycles, also using the G-MMS notion, while \citet{Suksompong19} focused on  paths and provided approximations of envy-freeness, proportionality, as well as another fairness notion called equitability.
\citet{IgarashiPe19} considered fairness in conjunction with the economic efficiency notion of Pareto optimality.
\citet{BouveretCeLe18} studied the problem of \emph{chore division}, where all items yield disutility to the agents, and gave complexity results on deciding the existence of envy-free, proportional, and equitable allocations for paths and stars.
\citet{BeiSu21} and \citet{IgarashiZw21} proposed similar models in which the resource is \emph{divisible} and forms the edges of a graph (as opposed to the vertices).
\citet{ElkindSeSu21-Graph} studied such models with respect to maximin share fairness.

Considering connected allocations can also be useful in settings where we are not interested in connectedness per se, or perhaps the goods do not even lie on any graph.
A technique that has received interest recently is to arrange the goods on a path and compute a connected allocation with respect to the path.
Variants of this technique have been used to devise algorithms that find a fair allocation using few queries \citep{OhPrSu19} or divide goods fairly among groups of agents \citep{SegalhaleviSu19,KyropoulouSuVo19}.

A related line of work also combines graphs with resource allocation, but uses graphs to capture the connection between agents instead of goods.
In particular, a graph specifies the acquaintance relationship among agents.
\citet{AbebeKlPa17} and \citet{BeiQiZh17} defined graph-based versions of envy-freeness and proportionality with divisible resources where agents only evaluate their shares relative to other agents with whom they are acquainted.
\citet{BeynierChGo18} and \citet{BredereckKaNi18} studied the graph-based version of envy-freeness with indivisible goods.
\citet{AzizBoCa18} introduced a number of fairness notions parameterized by the acquaintance graph.
In addition to graphs, other types of constraints that have been considered in the fair division literature include cardinality constraints~\citep{BiswasBa18}, matroid constraints~\citep{GourvesMo19,DrorFeSe21}, and separation constraints~\citep{ElkindSeSu21,ElkindSeSu21-Land}.\footnote{Constraints in fair division have recently been surveyed by \citet{Suksompong21}.}

Beyond resource allocation, the problem of partitioning a graph into connected subgraphs has been studied in other areas of discrete mathematics and theoretical computer science \citep{DyerFr85,VanthofPaWo09,PaulusmaVa11}.

\section{Preliminaries}
\label{sec:prelim}

Let $N=\{1,2,\dots,n\}$ denote the set of agents, and $M=\{1,2,\dots,m\}$ the set of goods.
There is a bijection between the goods in $M$ and the $m$ vertices of a connected undirected graph $G$; we will refer to goods and vertices interchangeably.
A \emph{bundle} is a subset of goods, and an \emph{allocation} is a partition of $M$ into $n$ bundles $(M_1,\dots,M_n)$ such that agent $i$ receives bundle $M_i$.
A bundle is called \emph{connected} if the goods in it form a connected subgraph of $G$, and an allocation or a partition is \emph{connected} if all of its bundles are connected. We assume in this paper that allocations are required to be connected.

Each agent $i$ has a nonnegative \emph{utility} $u_i(M')$ for each bundle $M'\subseteq M$, where we assume without loss of generality that $u_i(\emptyset)=0$ for all $i$.
For a good $g\in M$, we will use $u_i(\{g\})$ and $u_i(g)$ interchangeably.
We assume that utilities are \emph{additive}, i.e., $u(M')=\sum_{g\in M'}u(g)$ for all $M'\subseteq M$; this assumption is commonly made in the fair division literature, especially when studying maximin share fairness \citep{BouveretCeEl17,KurokawaPrWa18,GourvesMo19,LoncTr18}.
An \emph{instance} consists of the goods, their underlying graph, the agents, and their utilities for the goods.

We are ready to define maximin share fairness.

\begin{definition}
Given a graph $G$, an additive utility function $u$, and the number of agents $n$, the \emph{graph maximin share (G-MMS)} for $G,u,n$ is defined as
\[
\text{G-MMS}(G,u,n) := \max_{(M_1,\dots,M_n)}\min_{i=1,\dots,n} u(M_i),
\]
where the maximum is taken over all partitions $(M_1,\dots,M_n)$ that are connected with respect to $G$.
The \emph{maximin share (MMS)} for $u,n$ is defined as \[\text{MMS}(u,n) := \text{G-MMS}(K_m,u,n),\] where $K_m$ denotes the complete graph over the goods.
When the parameters are clear from the context, we will refer to the graph maximin share and the maximin share simply as G-MMS and MMS, respectively.
A partition for which the maximum is attained is called a \emph{G-MMS partition} (resp., \emph{MMS partition}).
\end{definition}

It follows from the definition that
\[\text{G-MMS}(G,u,n)\leq \text{MMS}(u,n)\leq \frac{u(M)}{n}\]
for all $G,u,n$, and $\text{G-MMS}(G_1,u,n)\leq \text{G-MMS}(G_2,u,n)$ if $G_1$ is a subgraph of $G_2$.
Moreover, $\text{G-MMS}(G,u,n)= \text{MMS}(u,n)=0$ if $m<n$.

Next, we define the price of connectivity.

\begin{definition}
\label{def:PoC}
Given a graph $G$ and the number of agents $n$, the \emph{price of connectivity (PoC)} of $G$ for $n$ agents is defined as
\[\sup_{u}\frac{\text{MMS}(u,n)}{\text{G-MMS}(G,u,n)},\]
where the supremum is taken over all possible additive utility functions $u$.\footnote{We interpret $\frac{0}{0}$ in this context to be equal to $1$. Note that $\text{MMS}(u,n) = 0$ if and only if $\text{G-MMS}(G,u,n) = 0$---indeed, since we assume that $G$ is connected, both conditions are equivalent to the condition that fewer than $n$ goods yield a positive utility according to $u$.} We denote the PoC of a graph $G$ for $n$ agents by $\text{PoC}(G,n)$.
\end{definition}

By definition of the PoC, we have \begin{equation}
\text{PoC}(G,n)\cdot \text{G-MMS}(G,u,n) \geq \text{MMS}(u,n)
\label{eq:PoC}
\end{equation} for any $G,u,n$, and the factor $\text{PoC}(G,n)$ cannot be replaced by any smaller factor.
When $G$ and $n$ are clear from the context, we will refer to $\text{PoC}(G,n)$ simply as PoC. Note that the PoC is always at least $1$, and is exactly $1$ for complete graphs of any size. Moreover, the PoC is $1$ if $m\leq n$.

Suppose that for some graph $G$ and number of agents $n$, there always exists a connected allocation that gives each agent at least $\beta$ times her G-MMS.
By \eqref{eq:PoC}, this allocation also gives each agent at least $\beta/\text{PoC}(G,n)$ times her MMS.
Prior work has established that $\beta = 1$ when $n=2$ and $G$ is arbitrary \cite[Cor.~2]{LoncTr18}, as well as when $G$ is a tree and $n$ is arbitrary \cite[Thm.~5.4]{BouveretCeEl17}.
Hence, in these cases, we can guarantee each agent at least $1/\text{PoC}(G,n)$ times her MMS.
The factor $1/\text{PoC}(G,n)$ is also the best possible.
To see this, consider $n$ agents with the same utility function $u$.
From the definition of G-MMS, any connected allocation gives some agent a value of at most $\text{G-MMS}(G,u,n)$.
By considering $u$ such that $\text{G-MMS}(G,u,n)$ is arbitrarily close to $\text{MMS}(u,n)/\text{PoC}(G,n)$, this agent receives arbitrarily close to $1/\text{PoC}(G,n)$ times her MMS.
To summarize, we have the following proposition.

\begin{proposition}
\label{prop:mms-CG-approx}
Let $n$ be any positive integer and $G$ be any graph.
If $n=2$ (and $G$ is arbitrary), or if $G$ is a tree (and $n$ is arbitrary), then there always exists a connected allocation that gives each agent at least $1/\text{PoC}(G,n)$ times her MMS. Moreover, the factor $1/\text{PoC}(G,n)$ is tight in both cases.
\end{proposition}

Proposition~\ref{prop:mms-CG-approx} implies that if there are two agents or $G$ is a tree, in order to determine the optimal MMS approximation for agents with possibly different utilities, it suffices to determine the value $\text{PoC}(G,n)$, which only concerns a single utility function.

We now introduce relaxations of envy-freeness \citep{LiptonMaMo04,CaragiannisKuMo16}.

\begin{definition}
An allocation $(M_1,\dots,M_n)$ satisfies
\begin{itemize}
\item \emph{envy-freeness up to $k$ goods (EF$k$)}, for a given nonnegative integer $k$, if for any agents $i,j$,
there exists a (possibly empty) bundle $M'\subseteq M_j$ with $|M'|\leq k$ such that $u_i(M_i)\geq u_i(M_j\setminus M')$.
\item \emph{envy-freeness up to any good (EFX)} if for any agents $i,j$ and any good $g\in M_j$, we have $u_i(M_i)\geq u_i(M_j\setminus\{g\})$.
\end{itemize}
\end{definition}

An EF0 allocation is said to be \emph{envy-free}. It follows immediately from the definition that envy-freeness implies EFX, which in turn implies EF1.
If we do not have to allocate all of the goods, achieving envy-freeness and all of its relaxations is trivial, for example by simply not allocating any good. Hence we will assume that all goods must be allocated when we discuss envy-freeness and its relaxations.

All graphs considered in this paper are assumed to be connected. The \emph{vertex connectivity} (or simply \emph{connectivity}) of a graph $G$ is the minimum number of vertices whose deletion disconnects $G$.
A graph with vertex connectivity at least $k$ is said to be \emph{$k$-connected}. By definition, every connected graph is 1-connected. A 2-connected graph is also called \emph{biconnected}.
A \emph{bipolar ordering} (also called \emph{bipolar numbering}) of a graph is an ordering of its vertices such that every prefix and every suffix of the ordering forms a connected graph.

\section{Maximin Share Fairness}
\label{sec:mms}

In this section, we consider maximin share fairness.
Our goal is to derive bounds on the PoC for arbitrary graphs in the case of two agents, and for paths and stars in the general case.
By Proposition~\ref{prop:mms-CG-approx}, this also yields the optimal MMS approximation for each of these cases.

\subsection{Two Agents}
\label{sec:mms-two}
We first focus on the case of two agents and start by establishing the PoC for all graphs with connectivity~$1$.

\begin{theorem}
\label{thm:mms-two-connectivity-1}
Let $G$ be a graph with connectivity exactly $1$, and let $k\geq 2$ be the maximum number of connected components that can result from deleting a single vertex of $G$. Then $\text{PoC}(G,2) = k$.
\end{theorem}
\begin{proof}
First, we show that the PoC of $G$ is at least $k$.
Let $v$ be a vertex of $G$ whose deletion results in $k$ components.
Consider a utility function with value $k$ for $v$, value $1$ for an arbitrary vertex in each of the $k$ components, and value $0$ for all other vertices.
The MMS is $k$. In any connected bipartition, the part that does not contain $v$ is a subset of one of the $k$ components, so this part has value at most $1$. Hence the PoC is at least $k$.

Next, we show that the PoC of $G$ is at most $k$. Take an arbitrary utility function $u$, and assume without loss of generality that $u(M)=1$.
Since $\text{MMS}(u,2)\leq u(M)/2 = 1/2$, the desired claim follows if there is a connected bipartition such that both parts have value at least $1/(2k)$.
Assume that no such bipartition exists.

Pick a spanning tree $T$ of $G$, and let $v$ be an arbitrary vertex. The removal of $v$ results in a number of subtrees of $T$; clearly, at most one of these subtrees can have value more than $1/2$.
If such a subtree exists, we move from $v$ towards the adjacent vertex in that subtree and repeat the procedure with the new center vertex.
Note that we will never traverse back an edge---otherwise there are two disjoint subtrees with value more than $1/2$ each, contradicting $u(M)=1$.
Since the tree is finite, we eventually reach a vertex $v$ such that all subtrees $T_1,\dots,T_r$ resulting from the removal of $v$ have value at most $1/2$ each.

Since $T_i$ and $T\setminus T_i$ are both connected for every $i$, by our earlier assumption, each of the subtrees $T_1,\dots,T_r$ has value less than $1/(2k)$.
Recall that in the original graph $G$, removing $v$ can result in at most $k$ components. This means that if $r>k$, the $r$ subtrees must be connected by some edges not belonging to $T$.
If subtrees $T_i$ and $T_j$ are connected by such an edge, we can merge $T_i$ and $T_j$ into one component.
Note that $T_i\cup T_j$ has value less than $1/(2k)+1/(2k) = 1/k \leq 1/2$, so since  $T_i\cup T_j$ and $T\setminus(T_i\cup T_j)$ are both connected, $T_i\cup T_j$ must again have value less than $1/(2k)$.
Our procedure can be repeated until the components can no longer be merged, at which point we are left with at most $k$ components.
Each of these components has value less than $1/(2k)$, which implies that $v$ has value more than $1-k/(2k) = 1/2$.
In this case, a bipartition with $v$ as one part is an MMS partition, so $\text{MMS}(u,2) = 1 - u(v)$.
On the other hand, at least one of the (at most) $k$ components has value at least $(1-u(v))/k$, which is $1/k$ of the MMS.
We can take a connected bipartition with such a component as one part and obtain the desired result.
\end{proof}

We remark that the proof of Theorem~\ref{thm:mms-two-connectivity-1} also yields a polynomial-time algorithm for computing a bipartition such that both parts have value at least $1/k$ of the MMS. To compute an allocation between two agents such that both agents receive $1/k$ of their MMS, we simply let the first agent compute a desirable bipartition, and let the second agent choose the part that she prefers. Since $\text{MMS}(u,2)\leq u(M)/2$, the second agent is always satisfied.

Before we move on to results about graphs with higher connectivity, we show the following lemma, which will help simplify our subsequent proofs. The lemma implies that in order to prove an upper bound on the PoC in the case of two agents, it suffices to establish the bound for utility functions such that in an MMS partition, the two parts are of equal value.

\begin{lemma}
\label{lem:mms-two-equal}
For $n=2$ and any graph $G$, the PoC remains the same if instead of taking the supremum in Definition~\ref{def:PoC}
\[\sup_{u}\frac{\text{MMS}(u,2)}{\text{G-MMS}(G,u,2)}\]
over all utility functions $u$, we only take the supremum over all utility functions $u$ such that in any MMS partition according to $u$, the two parts are of equal value.
\end{lemma}

\begin{proof}
Let $u$ be an arbitrary utility function, and suppose that in an MMS partition, the two parts are of value $x\leq y$. We have $\text{MMS}(u,2)=x$. Let $\alpha := \frac{\text{MMS}(u,2)}{\text{G-MMS}(G,u,2)}$.
In any connected bipartition, each part either has value at most $x/\alpha$, or at least $(x+y) - x/\alpha = y + (1-1/\alpha)x$.

Consider a modified utility function $u'$ where in the MMS partition above, we arbitrarily decrease the values of some goods in the part with value $y$ so that the part has value $x$.
It is clear that $\text{MMS}(u',2) = x$.
With respect to $u'$, in any connected bipartition, each part either has value at most $x/\alpha$, or at least $y + (1-1/\alpha)x - (y-x) = (2-1/\alpha)x$.
This means that $\text{G-MMS}(G,u',2)\leq x/\alpha = \text{MMS}(u',2)/\alpha$, or $\frac{\text{MMS}(u',2)}{\text{G-MMS}(G,u',2)}\geq \alpha$.
Since the two parts in any MMS partition according to $u'$ are of equal value, the proof is complete.
\end{proof}

Next, we consider biconnected graphs, i.e., graphs with connectivity at least $2$. We show that the PoC is at most $4/3$ for all such graphs---this is in contrast to graphs with connectivity $1$, which have PoC at least $2$ according to Theorem~\ref{thm:mms-two-connectivity-1}.
For this result, we will use a property of biconnected graphs which we state in the following proposition.
An \emph{open ear decomposition} of a graph consists of a cycle as the first \emph{ear} and a sequence of paths as subsequent ears such that in each path, the first and last vertices (which must be different) belong to previous ears while the remaining vertices do not.
See Figure~\ref{fig:ear} for an illustration.

\begin{figure}[!ht]
\centering
\begin{tikzpicture}[scale=0.95]
\draw [teal] (4.8,3.2) ellipse (2cm and 1.5cm);
\draw [red,fill=white] (4,2) ellipse (2cm and 1cm);
\draw [blue,fill=white] (3,3) ellipse (2cm and 1cm);
\draw [fill] (3.15,4) circle [radius=0.09];
\draw [fill] (4.97,2.86) circle [radius=0.09];
\draw [fill] (2.02,2.13) circle [radius=0.09];
\draw [fill] (3.62,2.05) circle [radius=0.09];
\draw [fill] (1.22,3.44) circle [radius=0.09];
\draw [fill] (6,2.04) circle [radius=0.09];
\draw [fill] (4.8,1.09) circle [radius=0.09];
\draw [fill] (3.3,1.07) circle [radius=0.09];
\draw [fill] (5,4.69) circle [radius=0.09];
\draw [fill] (6.62,3.8) circle [radius=0.09];
\node [blue] at (3,3) {$G_1$};
\node [red] at (4.75,1.8) {$G_2$};
\node [teal] at (5.6,3.8) {$G_3$};
\node at (1.0,3.7) {$g_1$};
\node at (2.95,4.3) {$h_1$};
\node at (4.65,2.87) {$g_2$};
\node at (3.62,1.75) {$h_2$};
\node at (1.75,1.95) {$h_3$};
\end{tikzpicture}
\caption{An example of an open ear decomposition with ears $G_1$, $G_2$, and $G_3$. The vertex labels are used in the proof of Theorem~\ref{thm:mms-two-connectivity-2}.}
\label{fig:ear}
\end{figure}

\begin{proposition}[\citep{Whitney32a,Whitney32b}]
\label{prop:biconnected}
In a biconnected graph with at least three vertices, any two vertices belong to a common cycle, and there exists an open ear decomposition.
Moreover, we may choose any cycle in the graph as the first ear.\footnote{There is also a linear-time algorithm for computing an open ear decomposition with an arbitrary cycle as the first ear \citep{Schmidt13}.}
\end{proposition}

\begin{theorem}
\label{thm:mms-two-connectivity-2}
Let $G$ be a biconnected graph. Then $\text{PoC}(G,2) \leq 4/3$.
\end{theorem}
\begin{proof}
The case $m\leq 2$ is trivial since $n=2$ and the PoC is $1$ in this case, so consider $m\geq 3$. Take an arbitrary utility function $u$, and assume without loss of generality that $u(M)=1$.
By Lemma~\ref{lem:mms-two-equal}, we may also assume that $\text{MMS}(u,2)=1/2$.
Call a good \emph{heavy} if it has value strictly more than $1/4$.
Since there can be at most one heavy good in each part of an MMS partition, there are at most two heavy goods in total.
Pick goods $g_1$ and $g_2$ so that together they include all of the heavy goods.
By Proposition~\ref{prop:biconnected}, there is a cycle in $G$ containing $g_1$ and $g_2$, and an open ear decomposition with this cycle as the first ear.

We will construct a bipolar ordering of the vertices that begins with $g_1$ and ends with $g_2$.
Assume that the first ear is a cycle with vertex order $$g_1,h_1,\dots,h_i,g_2,h_{i+1},\dots,h_j;$$
see Figure~\ref{fig:ear} for an illustration.
We arrange these vertices as $$g_1,h_1,h_2,\dots,h_i,h_j,h_{j-1},\dots,h_{i+1},g_2.$$
For each subsequent ear, suppose that the two vertices belonging to previous ears are $h$ and $h'$, where $h$ appears before $h'$ in the current ordering.
We insert the remaining vertices on the path from $h$ to $h'$ into the ordering directly after $h$, following the same order as in the path.
One can check (for example, by induction on the number of ears) that the resulting ordering is a bipolar ordering beginning with $g_1$ and ending with $g_2$.

Consider first the case where $\max\{u(g_1),u(g_2)\}> 1/2$; assume without loss of generality that $u(g_1)>1/2$. In this case, $\text{MMS}(u,2)=1-u(g_1)<1/2$, contradicting the assumption that $\text{MMS}(u,2)=1/2$.

Assume now that $\max\{u(g_1),u(g_2)\}\leq 1/2$, and recall that $u(g)\leq 1/4$ for all $g\not\in\{g_1,g_2\}$.
Since $\text{MMS}(u,2) = 1/2$, it suffices to find a connected bipartition such that both parts have value at least $3/8$.
Let $S=\{g_1\}$, so $u(S)\leq 1/2$. We add one good at a time to $S$ following the bipolar ordering until $u(S)\geq 1/2$. Since $u(g_2)\leq 1/2$, we stop (not necessarily directly) before we add $g_2$.
Moreover, since each good besides $g_1$ and $g_2$ has value at most $1/4$, at some point during this process we must have $3/8\leq u(S)\leq 5/8$. In the bipartition with $S$ as one part, both parts are connected and have value at least $3/8$, completing the proof.
\end{proof}

Unlike for Theorem~\ref{thm:mms-two-connectivity-1}, the proof of Theorem~\ref{thm:mms-two-connectivity-2} does not directly lead to a polynomial-time algorithm for computing an allocation such that both agents receive at least $3/4$ of their MMS.
The problematic step is when we apply Lemma~\ref{lem:mms-two-equal}, since computing the MMS value is NP-hard by a straightforward reduction from the partition problem.
\citet{Woeginger97} showed that a PTAS for the problem exists---using his PTAS, we can obtain a $(3/4-\epsilon)$-approximation algorithm that runs in polynomial time for any constant $\epsilon > 0$.
Nevertheless, we show in Appendix~\ref{app:algo-connectivity-2} that by building upon the proof of Theorem~\ref{thm:mms-two-connectivity-2}, we can also achieve a polynomial-time $3/4$-approximation algorithm.

In light of Theorems~\ref{thm:mms-two-connectivity-1} and \ref{thm:mms-two-connectivity-2}, it is tempting to believe that for graphs with connectivity $3$ or higher, the PoC is strictly less than $4/3$.
Perhaps surprisingly, this is not the case: a counterexample is the wheel graph shown in Figure~\ref{fig:wheel}, which has connectivity $3$.
In the instance shown in the figure, the MMS is $4$ while the G-MMS is $3$, so the PoC of the graph is at least $4/3$ (and by Theorem~\ref{thm:mms-two-connectivity-2}, exactly $4/3$).
The key point of this example is that the graph cannot be partitioned into two connected subgraphs in such a way that one subgraph contains the vertices with value $1$ and $3$, while the other subgraph contains the two vertices with value $2$.
This observation allows us to generalize the counterexample. A graph is said to be \emph{$2$-linked} if for any two disjoint pairs of vertices $(a,b)$ and $(c,d)$, there exist two vertex-disjoint paths, one from $a$ to $b$ and the other from $c$ to $d$.

\begin{figure}[!ht]
\centering
\begin{tikzpicture}[scale=0.85]
\draw (2,2) circle [radius = 2];
\draw [fill] (2,0) circle [radius = 0.1];
\draw [fill] (2,4) circle [radius = 0.1];
\draw [fill] (0,2) circle [radius = 0.1];
\draw [fill] (4,2) circle [radius = 0.1];
\draw [fill] (3.41,3.41) circle [radius = 0.1];
\draw [fill] (0.59,3.41) circle [radius = 0.1];
\draw [fill] (3.41,0.59) circle [radius = 0.1];
\draw [fill] (0.59,0.59) circle [radius = 0.1];
\draw [fill] (2,2) circle [radius = 0.1];
\draw (2,2) -- (2,0);
\draw (2,2) -- (2,4);
\draw (2,2) -- (0,2);
\draw (2,2) -- (4,2);
\draw (2,2) -- (3.41,3.41);
\draw (2,2) -- (0.59,3.41);
\draw (2,2) -- (3.41,0.59);
\draw (2,2) -- (0.59,0.59);
\node at (4.3,2) {$3$};
\node at (2,4.3) {$1$};
\node at (-0.3,2) {$0$};
\node at (2,-0.3) {$0$};
\node at (3.61,3.61) {$2$};
\node at (0.39,3.61) {$0$};
\node at (3.61,0.39) {$2$};
\node at (0.39,0.39) {$0$};
\node at (2.15,1.6) {$0$};
\end{tikzpicture}
\caption{An instance showing that the PoC of a wheel graph is at least $4/3$.}
\label{fig:wheel}
\end{figure}
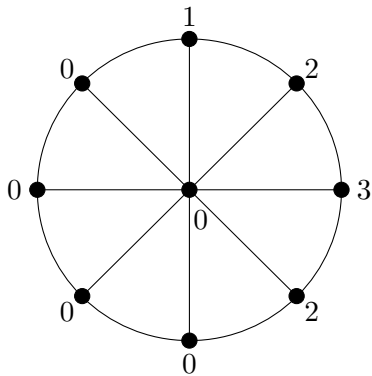

\begin{proposition}
\label{prop:mms-2-2-linked}
Let $G$ be a graph that is not $2$-linked. Then $\text{PoC}(G,2) \geq 4/3$.
\end{proposition}
\begin{proof}
Suppose that $G$ is not $2$-linked, and let $(a,b)$ and $(c,d)$ be disjoint pairs of vertices such that there do not exist two disjoint paths, one from $a$ to $b$ and the other from $c$ to $d$.
Consider a utility function $u$ such that $u(a)=u(b)=2$, $u(c)=3$, $u(d)=1$, and $u(g)=0$ for every other vertex $g$. We have $\text{MMS}(u,2) = 4$.
On the other hand, the graph cannot be partitioned into two connected subgraphs in such a way that one subgraph contains $a$ and $b$ while the other subgraph contains $c$ and $d$---indeed, such a partition would give rise to two disjoint paths that cannot exist by our assumption.
This means that $\text{G-MMS}(G,u,2) \leq 3$. Hence $\text{PoC}(G,2) \geq 4/3$.
\end{proof}

Every graph with connectivity at most $2$ is not $2$-linked,\footnote{Indeed, given such a graph, let $a,b$ be two vertices whose removal disconnects the graph, and let $c,d$ be vertices from distinct components in the resulting graph. Then any path between $c$ and $d$ must go through either $a$ or $b$.} and Figure~\ref{fig:wheel} shows an example of a $3$-connected graph that also does not satisfy the property.
In fact, \citet{Mezaros15} constructed a $5$-connected graph that still fails to be $2$-linked!\footnote{On the other hand, a $6$-connected graph is always $2$-linked \citep{Jung70}.} Combining these facts with Theorem~\ref{thm:mms-two-connectivity-2} yields the following corollaries:

\begin{corollary}
\label{cor:mms-two-connectivity-2}
For every graph $G$ with connectivity $2$, $\text{PoC}(G,2)=4/3$.
\end{corollary}

\begin{corollary}
For some graph $G$ with connectivity $5$, $\text{PoC}(G,2)=4/3$.
\end{corollary}

While we have not been able to precisely determine the PoC for all graphs with connectivity $3$ or above, we present a conjecture that, if settled in the affirmative, would complete the picture for the two-agent case.
Before we can describe the conjecture, we need the following generalization of 2-linkedness \citep{Mezaros15}:

\begin{definition}
\label{def:linkedness}
Given positive integers $a,b$, a graph $G$ is said to be \emph{$(a,b)$-linked} if for any disjoint set of vertices $M_1,M_2$ with $|M_1|=a$ and $|M_2|=b$, there exist disjoint connected subgraphs $G_1,G_2$ of $G$ such that $M_i$ is contained in $G_i$ for $i=1,2$.
\end{definition}

For example, $(2,1)$-linkedness is equivalent to biconnectivity,\footnote{To see this, first consider a graph $G$ that is not biconnected---suppose that removing a vertex~$x$ disconnects $G$. If $y$ and $z$ are vertices in different components of the resulting disconnected graph, then taking $M_1=\{y,z\}$ and $M_2=\{x\}$ yields a violation of Definition~\ref{def:linkedness}, meaning that $G$ is not $(2,1)$-linked. Conversely, suppose that $G$ is biconnected, and consider any disjoint set of vertices $M_1 = \{y,z\}$ and $M_2 = \{x\}$. By definition of biconnectivity, the graph $G$ remains connected upon the removal of $x$. Hence, we may take $G_2$ to be the subgraph induced only on $x$ and $G_1$ to be the subgraph induced on all vertices except $x$ in Definition~\ref{def:linkedness}. This implies that $G$ is $(2,1)$-linked.} while $(2,2)$-linked graphs correspond to what we have so far called $2$-linked graphs.
The new definition allows us to extend the lower bound from Proposition~\ref{prop:mms-2-2-linked}.

\begin{proposition}
\label{prop:mms-2-k-linked}
Let $k$ be a positive integer, and let $G$ be a graph that is not $(2,k)$-linked. Then $\text{PoC}(G,2)\geq 2k/(2k-1)$.
\end{proposition}

\begin{proof}
Suppose that $G$ is not $(2,k)$-linked, and let $\{a_1,a_2\}$ and $\{b_1,b_2,\dots,b_k\}$ be sets of vertices for which there do not exist disjoint connected subgraphs separating them.
Consider a utility function $u$ such that $u(a_1)=u(a_2)=k$, $u(b_1)=k+1$, $u(b_2)=u(b_3)=\dots=u(b_k)=1$, and $u(g)=0$ for every other vertex $g$.
We have $\text{MMS}(u,2) = 2k$.
On the other hand, the graph cannot be partitioned into two connected subgraphs in such a way that one subgraph contains $a_1,a_2$ while the other subgraph contains $b_1,b_2,\dots,b_k$. 
Since all vertex values are integers, this implies that $\text{G-MMS}(G,u,2)\leq 2k-1$. Hence $\text{PoC}(G,2)\geq 2k/(2k-1)$.
\end{proof}

Our conjecture is that for biconnected graphs, the PoC is exactly captured by $(2,k)$-linkedness:

\begin{conjecture}
\label{conj:mms-two}
Let $k\geq 2$ be an integer, and let $G$ be a graph that is $(2,k-1)$-linked but not $(2,k)$-linked.
Then $\text{PoC}(G,2) = 2k/(2k-1)$.
\end{conjecture}

The case $k=2$ of Conjecture~\ref{conj:mms-two} holds by Corollary~\ref{cor:mms-two-connectivity-2}.
We demonstrate next that the conjecture also holds for `almost-complete' graphs, i.e., for complete graphs with a nonempty matching removed.
These graphs have minimum degree $m-2$, where $m$ is the number of vertices (i.e., goods), and, with the exception of the graph $L_5$ that results from removing two disjoint edges from the complete graph $K_5$ (Figure~\ref{fig:L5}), are $(2,m-3)$-linked but not $(2,m-2)$-linked.
We show that the PoC of these graphs is always exactly $(2m-4)/(2m-5)$.
The exceptional graph $L_5$ is not $2$-linked, so Proposition~\ref{prop:mms-2-2-linked} (or alternatively, the utilities in Figure~\ref{fig:L5}) implies that its PoC is at least $4/3$ instead of $6/5$.
In fact, since the graph has connectivity $3$, Theorem~\ref{thm:mms-two-connectivity-2} tells us that its PoC is exactly $4/3$, thereby again confirming Conjecture~\ref{conj:mms-two}.

\begin{figure}[!ht]
\centering
\begin{tikzpicture}[scale=0.85]
\draw [fill] (2,4) circle [radius = 0.1];
\draw [fill] (3.9,2.62) circle [radius = 0.1];
\draw [fill] (0.1,2.62) circle [radius = 0.1];
\draw [fill] (3.18,0.38) circle [radius = 0.1];
\draw [fill] (0.82,0.38) circle [radius = 0.1];
\draw (2,4) -- (3.9,2.62);
\draw (2,4) -- (0.1,2.62);
\draw (2,4) -- (3.18,0.38);
\draw (2,4) -- (0.82,0.38);
\draw (3.9,2.62) -- (0.1,2.62);
\draw (3.18,0.38) -- (0.82,0.38);
\draw (3.9,2.62) -- (0.82,0.38);
\draw (0.1,2.62) -- (3.18,0.38);
\node at (2,4.3) {$0$};
\node at (4.2,2.62) {$2$};
\node at (-0.2,2.62) {$1$};
\node at (3.48,0.38) {$2$};
\node at (0.52,0.38) {$3$};
\end{tikzpicture}
\caption{Graph $L_5$ and utilities showing that its PoC is at least $4/3$.}
\label{fig:L5}
\end{figure}
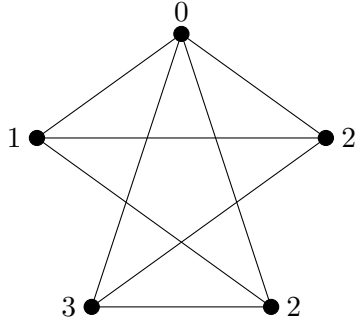

\begin{theorem}
\label{thm:mms-two-almost-complete}
Let $G$ be a graph that results from removing a nonempty matching from a complete graph with at least three vertices, and assume that $G$ is different from $L_5$.
Then $\text{PoC}(G,2)=(2m-4)/(2m-5)$.
\end{theorem}

To prove Theorem~\ref{thm:mms-two-almost-complete}, we will use the following lemma.

\begin{lemma}
\label{lem:subset-sum}
Let $k$ be a positive integer, $2\leq s\leq 2k$ be a real number, and let $x_1,x_2,\dots,x_k\geq 1$ be real numbers with sum $s$.
For any real number $0\leq r\leq s-2$, there exists a subset $J\subseteq\{1,2,\dots,k\}$ such that $r\leq\sum_{j\in J}x_j\leq r+2$.
\end{lemma}

\begin{proof}
We proceed by induction on $k$. For the base case $k=1$ we must have $s=2$, $x_1=2$, $r=0$, and the result holds trivially.
Suppose now that the result holds for $k-1$; we will prove it for $k$. Assume without loss of generality that $x_1=\max\{x_1,x_2,\dots,x_k\}$.

First, assume that $x_1\leq 2$. Define $y_i:=x_1+x_2+\dots+x_i$ for each $i$.
The sequence $0,y_1,y_2,\dots,y_k=s$ is strictly increasing and any two consecutive terms differ by at most $2$, so one of the terms $x_1+x_2+\dots+x_i$ must be between $r$ and $r+2$.
Hence we may take $J=\{1,2,\dots,i\}$ to fulfill the claim.

Assume from now on that $x_1>2$. We first prove the statement for $r\geq s/2 - 1$.
If $x_1 > s/2+1$, then since $x_i\geq 1$ for all $i$, we have
$$s = x_1+x_2+\dots+x_k>(s/2+1) + (k-1) = s/2+k,$$
thus $s>2k$, a contradiction.
So $x_1\leq s/2 + 1\leq r+2$. If $x_1\geq r$, we are done by choosing $J=\{1\}$, so assume that $x_1<r$.

Let $t := x_2+x_3+\dots+x_k$. Note that $0\leq t \leq s-2\leq 2(k-1)$ and $0<r-x_1\leq s-2-x_1 = t-2$.
Applying the induction hypothesis on $x_2,x_3,\dots,x_k$, we find that there is a set $L\subseteq\{2,3,\dots,k\}$ such that $r-x_1\leq\sum_{l\in L}x_l\leq r-x_1+2$.
Take $J=L\cup\{1\}$. We have $r\leq\sum_{j\in J}x_j\leq r+2$, as desired.

Finally, suppose that $r < s/2 - 1$. We have $$s-2\geq s-r-2 > s - (s/2-1)-2 = s/2-1,$$ so we know from the previous case ($r\geq s/2 - 1$) that there exists a subset $J\subseteq\{1,2,\dots,k\}$ for which $s-r-2\leq\sum_{j\in J}x_j\leq s-r$.
Since $\sum_{j=1}^k x_j = s$, it follows that $r\leq \sum_{j\in\{1,2,\dots,k\}\setminus J}x_j\leq r+2$, completing the proof.
\end{proof}

We are now ready to establish Theorem~\ref{thm:mms-two-almost-complete}.

\begin{proof}[Proof of Theorem~\ref{thm:mms-two-almost-complete}]
First, we show that the PoC of $G$ is at least $(2m-4)/(2m-5)$. Let $(v_1,v_2)$ be a missing edge.
Consider a utility function with value $m-2$ for each of $v_1$ and $v_2$, value $m-1$ for another vertex $v_3$, and value $1$ for each of the remaining $m-3$ vertices (so the total value is $4m-8$).
The MMS is $2m-4$, attained by the bipartition with $\{v_1,v_2\}$ as one part.
Take an arbitrary connected bipartition.
If $v_1$ and $v_2$ are in the same part, this part must contain at least one other vertex, so the other part has value at most $2m-5$.
On the other hand, if $v_1$ and $v_2$ are in different parts, the part that does not contain $v_3$ has value at most $2m-5$.
In either case, there is a part with value no more than $2m-5$, so the G-MMS is at most $2m-5$.
It follows that the PoC is at least $(2m-4)/(2m-5)$.

Next, we show that the PoC of $G$ is at most $(2m-4)/(2m-5)$.
Take an arbitrary utility function $u$, and assume without loss of generality that $u(M)=4m-8$.
By Lemma~\ref{lem:mms-two-equal}, we may also assume that $\text{MMS}(u,2) = (4m-8)/2 = 2m-4$. It suffices to show that $\text{G-MMS}(G,u,2)\geq 2m-5$.
Consider any MMS partition. If the partition is connected, we have that the G-MMS is $2m-4$.
Suppose therefore that the partition is not connected. Since $G$ results from removing a nonempty matching from a complete graph, this means that (at least) one of the parts corresponds to a missing edge.
Let $v_1$ and $v_2$ be the two vertices in that part (so $u(\{v_1,v_2\}) = 2m-4$), and $v_3,\dots,v_m$ be the remaining vertices of $G$.

Assume first that there exists a vertex $v\not\in\{v_1,v_2\}$ such that $u(v)\leq 1$. We have $2m-4\leq u(\{v_1,v_2,v\})\leq 2m-3$, and the vertices $v_1,v_2,v$ form a connected subgraph.
Moreover, since the graph $G$ is different from $L_5$, the remaining vertices also form a connected subgraph; together these vertices have value at least $(4m-8)-(2m-3) = 2m-5$.
Hence, in the connected bipartition with $\{v_1,v_2,v\}$ as one part, both parts have value at least $2m-5$.
It follows that $\text{G-MMS}(G,u,2)\geq 2m-5$ in this case.

Assume now that every vertex $v\not\in\{v_1,v_2\}$ satisfies $u(v) > 1$.
If $u(v_1) \geq 2m-5$, then taking the connected bipartition with $v_1$ alone as one part again yields $\text{G-MMS}(G,u,2)\geq 2m-5$; an analogous argument applies if $u(v_2) \geq 2m-5$.
Suppose therefore that $\max\{u(v_1),u(v_2)\} < 2m-5$.
Since $u(v_1)+u(v_2) = 2m-4$, we have $1 < u(v_1) < 2m-5$, and so $0 < 2m-5-u(v_1) < 2m-6$.
Applying Lemma~\ref{lem:subset-sum} with $k=m-2$, $s=2m-4$, $\{x_1,\dots,x_k\} = \{u(v_3),\dots,u(v_m)\}$, and $r = 2m-5-u(v_1)$, we find that there exists a subset of $\{u(v_3),\dots,u(v_m)\}$ for which the sum of the elements belongs to the interval $[2m-5-u(v_1), 2m-3-u(v_1)]$.
Letting $S$ be the set of corresponding goods along with $v_1$, we have $2m-5\leq u(S)\leq 2m-3$.
Hence, in the connected bipartition with $S$ as one part, both parts have value at least $2m-5$.
Therefore $\text{G-MMS}(G,u,2)\geq 2m-5$ in this case as well, and the proof is complete.
\end{proof}

One can check that any graph $G$ satisfying the condition of Theorem~\ref{thm:mms-two-almost-complete} is $(2,m-3)$-linked but not $(2,m-2)$-linked, so Theorem~\ref{thm:mms-two-almost-complete} confirms Conjecture~\ref{conj:mms-two} for this class of graphs.

\subsection{Any Number of Agents}
\label{sec:mms-any}

We proceed to the general setting where the goods are divided among an arbitrary number of agents.
In this setting, it is no longer true that the PoC alone captures the MMS approximation that can be guaranteed to the agents---this is evident in the case of a complete graph, where the PoC is $1$ by definition, but an allocation that gives all agents their full MMS does not always exist \citep{KurokawaPrWa18}.
At first glance, it may seem conceivable that certain graphs do not admit any useful MMS approximation.
However, we provide a non-trivial guarantee for arbitrary graphs that depends only on the number of agents and goods and, in particular, not on the utilities (Theorem~\ref{thm:mms-arbitrary-approx}). We begin by establishing a general upper bound on the PoC.

\begin{theorem}
\label{thm:mms-PoC-general}
For any graph $G$ and number of agents $n$, we have $\text{PoC}(G,n)\leq \max\{1,m-n+1\}$.
\end{theorem}

\begin{proof}
If $m<n$, the PoC is $1$. Assume that $m\geq n$, and consider an arbitrary utility function $u$. Let $(M_1,\dots,M_n)$ be a (not necessarily connected) partition of $M$ that maximizes $\min_{i=1,\dots,n} u(M_i)$. We assume without loss of generality that $|M_i|\geq 1$ for each $i$, which also means that $|M_i|\leq m-n+1$ for every $i$.

For each $i$, let $g_i$ be a good of highest value in $M_i$ according to $u$, and let $M_i'=\{g_i\}$.
As long as $\cup_{i=1}^n M_i'\neq M$, we add a good not already in $\cup_{i=1}^n M_i'$ to one of the bundles $M_i'$ so that the bundle remains connected; this is always possible since $G$ is connected.
At the end of this process, $(M_1',\dots,M_n')$ is a connected partition of $M$.
By our choice of $g_i$, we have
$$u(M_i')\geq\frac{1}{m-n+1}\cdot u(M_i)$$
for every $i$. It follows that
\begin{align*}
\text{G-MMS}(G,u,n)
&\geq\min_{i=1,\dots,n} u(M_i') \\
&\geq \frac{1}{m-n+1}\cdot \min_{i=1,\dots,n} u(M_i) \\
&= \frac{1}{m-n+1}\cdot \text{MMS}(u,n).
\end{align*}
Hence, we have that $\text{PoC}(G,n)\leq m-n+1$.
\end{proof}

As we will see in Theorems~\ref{thm:mms-any-star} and \ref{thm:mms-any-path}, the bound $m-n+1$ is tight for sufficiently short paths and all stars.
We now give a maximin share guarantee for arbitrary graphs.

\begin{theorem}
\label{thm:mms-arbitrary-approx}
For any graph $G$ and any number of agents $n$, if $m\ge n$, there exists a connected allocation that gives each agent at least $1/(m-n+1)$ of her MMS.
On the other hand, if $m < n$, there exists a connected allocation that gives each agent her full MMS.
\end{theorem}

\begin{proof}
The statement for $m < n$ is trivial since the MMS is $0$ in that case, so assume that $m\ge n$.
Take an arbitrary spanning tree $H$ of $G$. By Theorem~\ref{thm:mms-PoC-general}, $\text{PoC}(H,n)\leq m-n+1$. By Proposition~\ref{prop:mms-CG-approx}, there exists a connected allocation with respect to $H$ that gives each agent at least $1/(m-n+1)$ times her MMS. Since any connected allocation with respect to $H$ is also connected with respect to $G$, the conclusion follows.
\end{proof}

Next, we derive tight bounds on the PoC in the cases of paths and stars for any number of agents.
By Proposition~\ref{prop:mms-CG-approx}, this also yields the optimal MMS approximation for each of these cases.
The following simple fact will be useful:

\begin{lemma}
\label{lem:mms-inequality}
Let $m\geq n$, and let $M'\subseteq M$ be an arbitrary set of at least $m-n+1$ goods.
For an agent with utility function $u$, we have $u(M')\geq\text{MMS}(u,n)$.
\end{lemma}

\begin{proof}
Observe that in any partition of the vertices into $n$ parts, at least one of the parts is contained in $M'$.
In particular, this holds for an MMS partition.
It follows that $\text{MMS}(u,n)\leq u(M')$, as claimed.
\end{proof}

We begin with stars.

\begin{theorem}
\label{thm:mms-any-star}
Let $n\geq 2$ and let $G$ be a star. Then
$$\text{PoC}(G,n) =
\begin{cases}
m-n+1 & \text{ if } m \geq n;\\
1 & \text{ if } m < n.
\end{cases}
$$
Moreover, when $m\ge n$, there exists a polynomial-time algorithm that computes a connected allocation in which every agent receives at least $1/(m-n+1)$ of her MMS.
\end{theorem}

\begin{proof}
If $m<n$ the PoC is $1$, so assume that $m\geq n$.
We first show that the PoC is at least $m-n+1$.
Consider a utility function $u$ with value $m-n+1$ for the center vertex and for $n-2$ of the leaves, and value $1$ for each of the remaining $m-n+1$ leaves.
We have $\text{MMS}(u,n) = m-n+1$.
In any connected partition into $n$ parts, at least $n-1$ parts contain a single leaf.
This means that at least one of these parts contains a single leaf with value $1$.
Hence the PoC is at least $m-n+1$.

Next, we show that the PoC is at most $m-n+1$; while this bound already follows from Theorem~\ref{thm:mms-PoC-general}, our proof will yield a polynomial-time algorithm for computing a desirable connected allocation in the case of stars.
Take an arbitrary utility function $u$, let $v^*$ be the center vertex, and let $v_1,v_2,\dots,v_{n-1}$ be the leaves with the highest value where $u(v_1)\geq\dots\geq u(v_{n-1})$.
Consider a connected partition $\Pi$ with each of these $n-1$ vertices as a part, and the remaining $m-n+1$ vertices as the last part.

Let $A := M\setminus\{v^*,v_1,\dots,v_{n-2}\}$.
By Lemma~\ref{lem:mms-inequality},  $\text{MMS}(u,n)\leq u(A)$.
Since there are $m-n+1$ vertices in $A$ and $v_{n-1}$ is a vertex with the highest value, we have
$$u(v_{n-1})\geq \frac{1}{m-n+1}\cdot u(A)\geq \frac{1}{m-n+1}\cdot \text{MMS}(u,n).$$
It follows that $u(v_i)\geq \text{MMS}(u,n)/(m-n+1)$ for all $i=1,2,\dots,n-1$, so the first $n-1$ parts of $\Pi$ have value at least $\text{MMS}(u,n)/(m-n+1)$ each.
The last part of $\Pi$ is $B := M\setminus\{v_1,v_2,\dots,v_{n-1}\}$.
By Lemma~\ref{lem:mms-inequality} again, we have $\text{MMS}(u,n)\leq u(B)$.
This means that all parts of $\Pi$ have value at least $\text{MMS}(u,n)/(m-n+1)$, as desired.

This proof also gives rise to a polynomial-time algorithm for computing a connected allocation for $n$ agents on a star such that each agent receives at least $1/(m-n+1)$ of her MMS:
Let each of the first $n-1$ agents pick a favorite leaf from the remaining leaves in turn, and let the last agent take the remaining $m-n+1$ vertices.
\end{proof}

To address the more involved case of paths, we introduce an approximation of proportionality that can be of interest even in the absence of connectivity considerations.
Recall that an allocation is said to be \emph{proportional} if it gives every agent at least her \emph{proportional share}, which is defined as $u(M)/n$.
Even though a proportional allocation always exists for \emph{divisible} goods, as we explained in the introduction, this is not the case for indivisible goods---our definition of \emph{indivisible proportional share} therefore adapts proportionality to the setting of indivisible goods.
In order to ensure a nontrivial approximation, we will need to hypothetically remove up to $n-1$ goods from the entire bundle. 
Indeed, when there are $n-1$ goods overall, in any allocation, one of the agents is necessarily left empty-handed.
If this agent is only allowed to hypothetically remove at most $n-2$ goods, then she cannot guarantee any positive (multiplicative) approximation of her utility for the entire bundle.
Thus, we are interested in the optimal approximation of each agent's utility after $n-1$ goods are removed.
When the number of goods is large, this approximation is $1/n$, which is reasonable because there are $n$ agents.
However, for smaller numbers of goods, we will be able to achieve a better approximation, which is captured by our IPS factor in the following definition.

\begin{definition}
\label{def:indivisible-prop-share}
For positive integers $n,m$, define
$$\text{IPS}(n,m) =
\begin{cases}
\frac{1}{n} & \text{ if } m\geq 2n-1;\\
\frac{1}{m-n+1} & \text{ if } n\leq m< 2n-1;\\
0 & \text{ if } m < n.
\end{cases}
$$
Given $n$ agents and $m$ goods, a bundle $A$ is said to satisfy the \emph{indivisible proportional share (IPS) property} for an agent with utility function $u$ if there exists a (possibly empty) set $B\subseteq M\setminus A$ with $|B|\leq n-1$ such that
$$u(A)\geq \text{IPS}(n,m)\cdot u(M\setminus B).$$
An allocation is said to satisfy the IPS property if every agent receives a bundle that satisfies the IPS property.
For brevity, we will refer to a bundle or allocation that satisfies the IPS property as being IPS.
\end{definition}

We remark that IPS is a stronger property than PROP$^*(n-1)$ considered by \citet{SegalhaleviSu19}, which corresponds to taking $\text{IPS}(n,m) = 1/n$ for $m\geq n$ and $0$ for $m<n$.
(In particular, note that $\frac{1}{m-n+1} > \frac{1}{n}$ when $n\le m < 2n-1$.)
It is also stronger than PROP1 considered by \citet{ConitzerFrSh17} and \citet{AzizCaIg19}, as well as a proportionality relaxation studied by \citet{Suksompong19}.
Despite its strength, we show that an IPS allocation always exists. Moreover, we can obtain a connected IPS allocation if the graph is a path.

\begin{proposition}
\label{prop:indivisible-prop-share}
Let $n\geq 2$ and let $G$ be a path. There exists a connected IPS allocation of the $m$ goods to the $n$ agents.
\end{proposition}

\begin{proof}
If $m<n$, each agent needs utility $0$ in an IPS allocation, so the claim holds trivially. Assume that $m\geq n$.
Starting with an empty bundle, we process the goods along the path (say, from left to right) and add them one at a time to the current bundle until the bundle is IPS to at least one of the agents.
We then allocate the bundle to one such agent, and repeat the procedure with the remaining goods and agents. Any leftover goods are allocated to the agent who receives the last bundle.

We claim that this procedure always results in an IPS allocation.
Notice from Definition~\ref{def:indivisible-prop-share} that if a bundle is IPS for an agent, then so is any superset of the bundle.
Hence it suffices to show that after $n-1$ bundles are allocated, the last agent still finds the remaining bundle to be IPS.
Assume without loss of generality that the bundles are allocated to agents $1,2,\dots,n$ in this order, and let $u$ be the utility function of agent $n$.
The claim holds trivially if the empty bundle is IPS for agent $n$, so assume that it is not.
For $1\leq i\leq n-1$, let the bundle allocated to agent $i$ be $M_i = X_i\cup Y_i$, where $Y_i$ consists of the last good added to $M_i$ (if $M_i$ is nonempty), and $X_i$ consists of the remaining goods.
Let $X = \cup_{i=1}^{n-1}X_i$ and $Y = \cup_{i=1}^{n-1}Y_i$.
In particular, $|Y|\leq n-1$.

Let $M_n$ be the bundle allocated to agent $n$. 

\begin{itemize}
\item Case 1: $m\geq 2n-1$.
By definition of the procedure, agent $n$ does not find any of the bundles $X_1,\dots,X_{n-1}$ to be IPS.
In particular, noting that $Y\subseteq M\setminus X_i$ for each $1\leq i\leq n-1$ and taking $B=Y$ in Definition~\ref{def:indivisible-prop-share},
we have
$u(X_i) < \text{IPS}(n,m)\cdot u(M\setminus Y) = u(M\setminus Y)/n$
for all $i$.
Hence,
\begin{align*}
u(M_n)
&= u(M) - \sum_{i=1}^{n-1}u(X_i) - \sum_{i=1}^{n-1}u(Y_i) \\
&> u(M) - \frac{n-1}{n}\cdot u(M\setminus Y) - u(Y) \\
&= \frac{1}{n}\cdot u(M\setminus Y).
\end{align*}
Since $Y\subseteq M\setminus M_n$, bundle $M_n$ is IPS for agent $n$.

\item Case 2: $n\leq m\leq 2n-1$.
First, we show that at most $m-n$ of the first $n-1$ agents can receive at least two goods.
Assume for contradiction that at least $m-n+1$ of these agents receive at least two goods, and suppose that the first $m-n+1$ of them are agents $a_1,\dots,a_{m-n+1}$ in this order.
Let $j$ be the first good in agent $a_{m-n+1}$'s bundle.
We claim that the bundle consisting of good $j$ alone is IPS for agent $n$; this is sufficient for the desired contradiction because agent $n$ should have taken this bundle ahead of agent $a_{m-n+1}$.

Before agent $a_{m-n+1}$ receives her bundle, the goods in $X$ allocated to earlier agents are precisely those in the set $X' := \cup_{i=1}^{m-n} X_{a_i}$.
Let $Z=M\setminus (X'\cup\{j\})$.
Since $|X'|\geq m-n$, we have $|Z|\leq m - (m-n) - 1 = n-1$.
By definition of the procedure, agent $n$ does not find any of the bundles $X_{a_1},\dots,X_{a_{m-n}}$ to be IPS.
In particular, noting that $Z\subseteq M\setminus X_{a_i}$ and taking $B=Z$ in Definition~\ref{def:indivisible-prop-share},
we have $u(X_{a_i}) < u(M\setminus Z)/(m-n+1)$ for all $1\leq i\leq m-n$. Hence,
\begin{align*}
u(\{j\})
&= u(M) - u(X') - u(Z) \\
&= u(M\setminus Z) - \sum_{i=1}^{m-n}u\left(X_{a_i}\right)  \\
&> u(M\setminus Z) - \frac{m-n}{m-n+1}\cdot u(M\setminus Z) \\
&= \frac{1}{m-n+1}\cdot u(M\setminus Z).
\end{align*}
Since $Z\subseteq M\setminus\{j\}$, bundle $\{j\}$ is IPS for agent $n$, so agent $n$ should indeed have taken this bundle ahead of agent $a_{m-n+1}$. This contradiction means that at most $m-n$ of the first $n-1$ agents can receive at least two goods.

We now proceed in a similar way as in Case~1.
By definition of the procedure, agent $n$ does not find any of the bundles $X_1,\dots,X_{n-1}$ to be IPS.
In particular, noting that $Y\subseteq M\setminus X_i$ for each $1\leq i\leq n-1$ and taking $B=Y$ in Definition~\ref{def:indivisible-prop-share},
we have
$u(X_i) < \text{IPS}(n,m)\cdot u(M\setminus Y) = u(M\setminus Y)/(m-n+1)$
for all $i$.
Hence,
\begin{align*}
u(M_n)
&= u(M) - \sum_{i=1}^{n-1}u(X_i) - \sum_{i=1}^{n-1}u(Y_i) \\
&> u(M) - \frac{m-n}{m-n+1}\cdot u(M\setminus Y) - u(Y) \\
&= \frac{1}{m-n+1}\cdot u(M\setminus Y),
\end{align*}
where the inequality holds because at most $m-n$ of the sets $X_i$ are nonempty.
Since $Y\subseteq M\setminus M_n$, bundle $M_n$ is IPS for agent $n$. 
\end{itemize}
The two cases together complete the proof.
\end{proof}

Proposition~\ref{prop:indivisible-prop-share} allows us to establish the PoC for paths, which we do next in Theorem~\ref{thm:mms-any-path}. Conversely, the instances that we use to show the upper bound on the PoC in Theorem~\ref{thm:mms-any-path} also show that the factor $\text{IPS}(n,m)$ in the existence guarantee of Proposition~\ref{prop:indivisible-prop-share} cannot be improved.

\begin{theorem}
\label{thm:mms-any-path}
Let $n\geq 2$ and let $G$ be a path. Then
$$\text{PoC}(G,n) =
\begin{cases}
n & \text{ if } m\geq 2n-1;\\
m-n+1 & \text{ if } n\leq m< 2n-1;\\
1 & \text{ if } m < n.
\end{cases}
$$
\end{theorem}

\begin{proof}
If $m<n$ the PoC is $1$, so assume that $m\geq n$.
We will show that $\text{PoC}(G,n)=1/\text{IPS}(n,m)$.

First, we show that $\text{PoC}(G,n)\leq 1/\text{IPS}(n,m)$.
Take an arbitrary utility function $u$.
Applying Proposition~\ref{prop:indivisible-prop-share} to $n$ agents who have the same utility function $u$, we find that there exists a connected IPS allocation.
This means each agent $i$ receives a bundle $M_i$ for which there exists a set $B_i\subseteq M\setminus M_i$ with $|B_i|\leq n-1$ such that $u(M_i)\geq \text{IPS}(n,m)\cdot u(M\setminus B_i)$.
Since $|M\setminus B_i|\geq m-n+1$, Lemma~\ref{lem:mms-inequality} implies that $u(M\setminus B_i)\geq \text{MMS}(u,n)$.
Consequently, we have
$$u(M_i)\geq \text{IPS}(n,m)\cdot u(M\setminus B_i) \geq \text{IPS}(n,m)\cdot \text{MMS}(u,n)$$
for all agents $i$.
Hence $(M_1,\dots,M_n)$ is a connected partition with each part having value at least $\text{IPS}(n,m)\cdot \text{MMS}(u,n)$. 
It follows that $\text{PoC}(G,n)\leq 1/\text{IPS}(n,m)$.

Next, we show that $\text{PoC}(G,n)\geq 1/\text{IPS}(n,m)$. We consider two cases.

\begin{itemize}
\item Case 1: $m\geq 2n-1$.
Consider a utility function $u$ with value $1,n,1,\dots,n,1$ for the first $2n-1$ vertices on the path (so exactly $n$ vertices have value $1$), and value $0$ for the remaining vertices.
We have $\text{MMS}(u,n) = n$.
On the other hand, one can check that in any connected partition into $n$ parts, at least one of the parts has value at most $1$. Hence the PoC is at least $n = 1/\text{IPS}(n,m)$.
\item Case 2: $n\leq m< 2n-1$.
Consider a utility function $u$ with value $1,m-n+1,1,\dots,m-n+1,1$ for the first $2m-2n+1$ vertices on the path (so $m-n+1$ vertices have value $1$ while $m-n$ vertices have value $m-n+1$),
and value $m-n+1$ for the remaining $2n-1-m$ vertices.
In total, $n-1$ vertices have value $m-n+1$, and $m-n+1$ vertices have value $1$.
We have $\text{MMS}(u,n) = m-n+1$.
On the other hand, one can check that in any connected partition into $n$ parts, at least one of the parts has value at most $1$. Hence the PoC is at least $m-n+1 = 1/\text{IPS}(n,m)$.
\end{itemize}
In both cases we have $\text{PoC}(G,n)\geq 1/\text{IPS}(n,m)$, completing the proof.
\end{proof}

Note that in order to compute a connected allocation for $n$ agents on a path such that every agent receives at least a $1/\text{PoC}(G,n)=\text{IPS}(n,m)$ fraction of their MMS, we can use the algorithm in Proposition~\ref{prop:indivisible-prop-share}, which runs in polynomial time, to compute a connected IPS allocation.
The first part in the proof of Theorem~\ref{thm:mms-any-path} implies that this allocation fulfills the desired guarantee.

\section{Envy-Freeness Relaxations}
\label{sec:envyfree}
Having extensively studied maximin share guarantees in the presence of connectivity requirements in the previous section, we now turn our attention to relaxations of envy-freeness. 
We again determine the price that we have to pay in order to maintain connectivity---intuitively, the less connected the graph is, the higher this price becomes.
Unless specified otherwise, we allow agents to have arbitrary monotonic utilities in this section.

We say that a graph $G$ {\em guarantees} EF$k$ for $n$ agents if for all permitted utilities of the $n$ agents, there exists a connected EF$k$ allocation.

\subsection{Two Agents}
\label{sec:envyfree-two}
For two agents, \citet{BiloCaFl19} characterized the set of graphs that always admit an EF1 allocation regardless of the agents' utilities.
Their characterization is based on the observation that such graphs necessarily admit a vertex ordering to which a discrete variant of the cut-and-choose protocol can be applied---in other words, the ordering is bipolar.
The family of graphs that admit a bipolar ordering can be characterized using the block decomposition of a graph. A \emph{block} is a maximal biconnected subgraph of a graph, and a \emph{cut vertex} is a vertex whose removal increases the number of connected components in the graph.
The \emph{block decomposition} of a graph $G$ is a bipartite graph $B(G)$ with all blocks of $G$ on one side and all cut vertices of $G$ on the other side; there is an edge between a block and a cut vertex in $B(G)$ if and only if the cut vertex belongs to the block in $G$.\footnote{We refer to the paper of Bil\`{o} et al.~for examples of graphs and their block decompositions.} 

\begin{proposition}[\citep{BondyMu08}]
\label{prop:block-decomposition}
For any connected graph $G$, each pair of blocks share no edge and at most one cut vertex, and the block decomposition $B(G)$ is a tree.
\end{proposition}

Bil\`{o} et al.~showed that a connected graph $G$ guarantees EF1 for two agents if and only if a bipolar ordering exists in $G$, i.e., the blocks of $G$ can be arranged into a path.

\begin{proposition}[\citep{BiloCaFl19}]
\label{prop:EF-two-bipolar}
For any connected graph $G$, the following four conditions are equivalent:\footnote{Bil\`{o} et al.~used a slightly stronger definition of EF1 that they called ``envy-freeness up to one \emph{outer} good''.
In their definition, one is only allowed to remove a good if doing so leaves the remaining bundle connected.
It can be verified that their result also holds for the standard definition of EF1.
}
\begin{enumerate}
    \item[$(1)$] The block decomposition $B(G)$ is a path;
    \item[$(2)$] $G$ admits a bipolar ordering;
    \item[$(3)$] $G$ guarantees EF1 for two agents with arbitrary monotonic utilities;
    \item[$(4)$] $G$ guarantees EF1 for two agents with identical binary utilities.
\end{enumerate}
\end{proposition}

Bil\`{o} et al.'s characterization allows us to identify graphs for which an EF1 allocation always exists in the case of two agents.
However, for the remaining graphs, it does not provide any fairness guarantee. Our next result generalizes their characterization by giving the best possible EF$k$ guarantee that can be made for each specific graph.
In particular, we will show that a graph $G$ guarantees EF$k$ for two agents if and only if $G$ admits a bipolar ordering over a subset of the vertices where each vertex in the ordering has at most $k-1$ vertices `hanging' from it.

To formalize this idea, it will be useful to define the following notions.
Given a path $P$ in the block graph $B(G)$ of a graph $G$, for any vertex $v$ of $G$ that is not contained in any block in $P$, we define its \emph{guardian} to be the cut vertex $v'$ closest to $v$ in $B(G)$ that belongs to some block in $P$ (see Figure~\ref{fig:guardian} for an example); we say that $v$ is a \emph{dependent} of $v'$.
For a given graph, we define a \emph{merge} on a subset $V$ of vertices forming a connected subgraph to be an operation where we replace the vertices in $V$ by a single vertex $v$, and there is an edge between $v$ and another vertex $w$ in the new graph exactly when $w$ is adjacent to at least one vertex of $V$ in the original graph.
A path in a tree is said to be \emph{maximal} if each of its end vertices is a leaf of the tree.

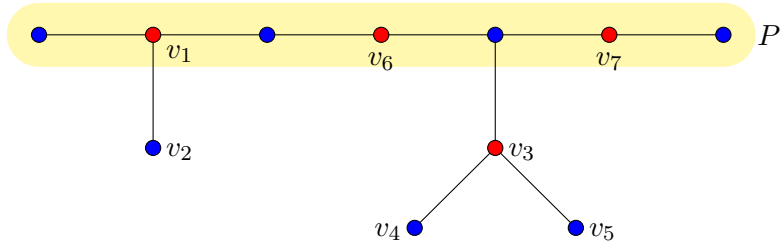
\begin{figure}[!ht]
\centering
\begin{tikzpicture}[scale=0.85]
\draw[line cap=round, yellow!40, line width=8.5mm] (0,4)--(9,4);
\draw (0,4) -- (9,4);
\draw (1.5,2.5) -- (1.5,4);
\draw (6,2.5) -- (6,4);
\draw (6,2.5) -- (4.94,1.44);
\draw (6,2.5) -- (7.06,1.44);
\draw [fill=blue] (0,4) circle [radius = 0.1];
\draw [fill=white] (1.5,4) circle [radius = 0.1];
\draw [fill=blue] (3,4) circle [radius = 0.1];
\draw [fill=white] (4.5,4) circle [radius = 0.1];
\draw [fill=blue] (6,4) circle [radius = 0.1];
\draw [fill=white] (7.5,4) circle [radius = 0.1];
\draw [fill=blue] (9,4) circle [radius = 0.1];
\draw [fill=blue] (1.5,2.5) circle [radius = 0.1];
\draw [fill=white] (6,2.5) circle [radius = 0.1];
\draw [fill=blue] (4.94,1.44) circle [radius = 0.1];
\draw [fill=blue] (7.06,1.44) circle [radius = 0.1];
\node at (1.85,3.75) {$v_1$};
\node at (1.85,2.45) {$v_2$};
\node at (6.35,2.45) {$v_3$};
\node at (4.59,1.39) {$v_4$};
\node at (7.41,1.39) {$v_5$};
\node at (4.5,3.65) {$v_6$};
\node at (7.5,3.65) {$v_7$};
\node at (9.6,4) {$P$};
\end{tikzpicture}
\caption{An example of a block decomposition $B(G)$ in the proof of Theorem~\ref{thm:EF-two-EFk}.
Blue vertices correspond to blocks in $G$ and white vertices correspond to cut vertices in $G$.
Here, $C(P)=\{v_1,v_3,v_6,v_7\}$.
In this example, $v_1$ is the guardian of all vertices in block $v_2$ except itself, $v_3$ is the guardian of all vertices in blocks $v_4$ and $v_5$ except itself, while $v_6$ and $v_7$ are not guardians of any vertices.
}
\label{fig:guardian}
\end{figure}

\begin{theorem}
\label{thm:EF-two-EFk}
For any connected graph $G$ and positive integer $k$, the following four conditions are equivalent:
\begin{enumerate}
\item[$(1)$] There exists a path $P$ in the block decomposition $B(G)$ such that each cut vertex that belongs to some block in $P$ has at most $k-1$ dependents;
\item[$(2)$] The vertices of $G$ can be partitioned into disjoint subsets $V_1,\ldots,V_r$ such that each $V_j$ forms a connected subgraph of size at most $k$ in $G$, and if we merge the vertices in every set $V_j$ separately, the resulting graph admits a bipolar ordering;
\item[$(3)$] $G$ guarantees EF$k$ for two agents with arbitrary monotonic utilities;
\item[$(4)$] $G$ guarantees EF$k$ for two agents with identical binary utilities.
\end{enumerate}
\end{theorem}

\begin{proof}
Consider the block decomposition $B(G)$, and recall from Proposition~\ref{prop:block-decomposition} that $B(G)$ is a tree. For each path $P$ of $B(G)$, we denote by $C(P)$ the set of cut vertices that belong to some block in $P$.

To show $(1) \Rightarrow (2)$, suppose that there exists a path $P$ in the block decomposition $B(G)$ such that each cut vertex in $C(P)$ has at most $k-1$ dependents. Take each set $V_j$ in the theorem statement to consist of a vertex in $C(P)$ along with all of its dependents. Clearly, at most $k$ vertices belong to each $V_j$. Also, each $V_j$ is connected since the vertices in $V_j$ form a connected subgraph of the block decomposition. Let $G'$ be the graph resulting from the merge operations on each $V_j$ separately. The block decomposition of $G'$ is a path, and hence $G'$ admits a bipolar ordering by Proposition~\ref{prop:EF-two-bipolar}.

To show $(2) \Rightarrow (3)$, suppose that the vertices of $G$ can be partitioned into disjoint subsets $V_1,\dots,V_r$ as defined in the statement of the theorem.
We will show that $G$ guarantees EF$k$ for two agents.
Consider arbitrary monotonic utilities of the two agents $u_i$ for $i=1,2$. Let $G'$ be the graph resulting from the merge operations on each $V_j$ for $j=1,\dots,r$.
We define the utility functions $u'_i$ on $G'$ for $i=1,2$, where the value of an agent for each bundle $M'$ is equal to her value for all vertices of $G$ that are merged into the vertices of $M'$. Specifically, for each $i=1,2$ and each bundle $M'$ in $G'$,
$$
u'_i(M')=u_i\left(\bigcup_{V_j \in M'}V_j\right).
$$
Note that each $u'_i$ remains monotonic and, by our assumption, $G'$ admits a bipolar ordering.
Thus, by Proposition~\ref{prop:EF-two-bipolar}, $G'$ admits a connected EF1 allocation $(M'_1,M'_2)$ with the utilities $u'_i$, so an agent's envy can be eliminated by removing a vertex of $G'$ from the other agent's bundle.
Consider the corresponding allocation $(M_1,M_2)$ of $G$, where $M_i=\bigcup_{V_j \in M'_i}V_j$ for $i=1,2$.
Since each vertex of $G'$ is a merge of at most $k$ vertices, any envy that results from this allocation can be eliminated by removing at most $k$ vertices, and so the allocation $(M_1,M_2)$ is a connected EF$k$ allocation of $G$.

The implication $(3) \Rightarrow (4)$ is immediate. To show $(4) \Rightarrow (1)$, suppose that for every path $P$ of $B(G)$, there exists some cut vertex in $C(P)$ with at least $k$ dependents.
We will show that there exist identical binary utility functions for which the graph $G$ does not admit an EF$k$ allocation.

Let $k^*\geq 1$ be the smallest number for which there exists a \emph{maximal} path $P$ in $B(G)$ such that each cut vertex in $C(P)$ is the guardian of at most $k^*-1$ vertices in $G$.
Choose a maximal path $P$ in $B(G)$ where each cut vertex in $C(P)$ has at most $k^*-1$ dependents; if several such paths exist, choose one that minimizes the number of vertices in $C(P)$ with exactly $k^*-1$ dependents.
By definition of $k^*$, we have $k^*-1 \geq k$.
Hence it suffices to show the existence of identical binary utility functions for which the graph $G$ does not admit an EF$(k^*-1)$ allocation.

Let $v\in C(P)$ be a cut vertex with $k^*-1$ dependents. It could be that $v$ is on the path $P$ itself (e.g., vertices $v_1$, $v_6$, and $v_7$ in Figure~\ref{fig:guardian}), or $v$ is not on the path $P$ but belongs to some block in $P$ (e.g., vertex $v_3$ in Figure~\ref{fig:guardian}).
We consider the two cases separately.

\begin{itemize}

\begin{figure}[!ht]
\centering
\subfloat{
\begin{tikzpicture}[scale=0.85]
\draw[line cap=round, yellow!40, line width=8.5mm] (0,4)--(9,4);
\draw[line cap=round, blue!20, line width=5mm] (0,4)--(3,4);
\draw[line cap=round, green!30, line width=5mm] (4.5,2.5)--(3.44,1.44)--(3.44,-0.06);
\draw[line cap=round, green!30, line width=5mm] (4.5,2.5)--(5.56,1.44)--(5.56,-0.06);
\draw[line cap=round, orange!40, line width=5mm] (6,4)--(9,4);
\draw[line cap=round, orange!40, line width=5mm] (7.5,4)--(7.5,2.5);
\draw (0,4) -- (9,4);
\draw (4.5,2.5) -- (4.5,4);
\draw (7.5,2.5) -- (7.5,4);
\draw (4.5,2.5) -- (3.44,1.44);
\draw (4.5,2.5) -- (5.56,1.44);
\draw (3.44,1.44) -- (3.44,-0.06);
\draw (5.56,1.44) -- (5.56,-0.06);
\draw (6,2.5) -- (4.5,4);
\draw [fill=blue] (0,4) circle [radius = 0.1];
\draw [fill=white] (1.5,4) circle [radius = 0.1];
\draw [fill=blue] (3,4) circle [radius = 0.1];
\draw [fill=white] (4.5,4) circle [radius = 0.1];
\draw [fill=blue] (6,4) circle [radius = 0.1];
\draw [fill=white] (7.5,4) circle [radius = 0.1];
\draw [fill=blue] (9,4) circle [radius = 0.1];
\draw [fill=blue] (4.5,2.5) circle [radius = 0.1];
\draw [fill=white] (3.44,1.44) circle [radius = 0.1];
\draw [fill=white] (5.56,1.44) circle [radius = 0.1];
\draw [fill=blue] (3.44,-0.06) circle [radius = 0.1];
\draw [fill=blue] (5.56,-0.06) circle [radius = 0.1];
\draw [fill=blue] (7.5,2.5) circle [radius = 0.1];
\draw [fill=blue] (6,2.5) circle [radius = 0.1];
\node at (9.6,4) {$P$};
\node at (1.5,3.5) {$L_v$};
\node at (2.95,0.75) {$T$};
\node at (8.1,3.1) {$R_v$};
\node at (4.33,3.75) {$v$};
\end{tikzpicture}
}
\vspace{5mm}
\subfloat{
\begin{tikzpicture}[scale=0.85]
\draw[line cap=round, yellow!40, line width=8.5mm] (0,4)--(9,4);
\draw[line cap=round, green!30, line width=5mm] (0,4)--(3,4);
\draw[line cap=round, green!30, line width=5mm] (3,4)--(3,1);
\draw[line cap=round, blue!20, line width=5mm] (4.5,2.5)--(4.5,-0.5);
\draw[line cap=round, orange!40, line width=5mm] (6,4)--(9,4);
\draw[line cap=round, orange!40, line width=5mm] (7.5,4)--(7.5,2.5);
\draw (0,4) -- (9,4);
\draw (4.5,-0.5) -- (4.5,4);
\draw (7.5,2.5) -- (7.5,4);
\draw (3,1) -- (3,4);
\draw (6,2.5) -- (4.5,4);
\draw [fill=blue] (0,4) circle [radius = 0.1];
\draw [fill=white] (1.5,4) circle [radius = 0.1];
\draw [fill=blue] (3,4) circle [radius = 0.1];
\draw [fill=white] (4.5,4) circle [radius = 0.1];
\draw [fill=blue] (6,4) circle [radius = 0.1];
\draw [fill=white] (7.5,4) circle [radius = 0.1];
\draw [fill=blue] (9,4) circle [radius = 0.1];
\draw [fill=blue] (4.5,2.5) circle [radius = 0.1];
\draw [fill=white] (4.5,1) circle [radius = 0.1];
\draw [fill=blue] (4.5,-0.5) circle [radius = 0.1];
\draw [fill=white] (3,2.5) circle [radius = 0.1];
\draw [fill=blue] (3,1) circle [radius = 0.1];
\draw [fill=blue] (7.5,2.5) circle [radius = 0.1];
\draw [fill=blue] (6,2.5) circle [radius = 0.1];
\node at (9.6,4) {$P'$};
\node at (5.05,0.25) {$L_v$};
\node at (2.55,1.75) {$T$};
\node at (8.1,3.1) {$R_v$};
\node at (4.33,3.75) {$v$};
\end{tikzpicture}
}
\caption{An example of a switch operation on the tree $B(G)$ in Case~1 of the proof of Theorem~\ref{thm:EF-two-EFk}.
The top and bottom figures are the trees before and after the operation, respectively.}
\label{fig:EFk-case1}
\end{figure}

\item Case 1: $v$ is in the path $P$ itself.
Let $L_v$ and $R_v$ be the subtree of the tree $B(G)$ rooted at $v$ starting with each of the two blocks adjacent to $v$ on the path $P$, respectively.
For each subtree besides $L_v$ and $R_v$ of the tree $B(G)$ rooted at $v$, with a block adjacent to $v$ in $B(G)$ as the root of the subtree, define its \emph{size} to be the number of dependents of $v$ in $G$ belonging to at least one block in the subtree. Note that the size can be different from the number of vertices in the subtree in $B(G)$---for example, if we take $v = v_1$ in Figure~\ref{fig:guardian}, the subtree with vertex $v_2$ as the root contains only one vertex in $B(G)$ (i.e., $v_2$), but this vertex may represent several vertices in $G$.

Suppose that $T$ is a largest subtree among such subtrees and has size $r\leq k^*-1$.
We claim that at least $r$ vertices of $G$ (excluding $v$) belong to some block in $L_v$.
Assume for contradiction that there are at most $r-1$ such vertices.
In $B(G)$, we switch $L_v$ with $T$ and choose an arbitrary path of $T$ that contains a leaf of $B(G)$ to be on the main path $P$ (see Figure~\ref{fig:EFk-case1}). Let $P'$ denote the new maximal path.
Since $v$ loses at least $r$ dependents and gains at most $r-1$ new dependents, $v$ now has at most $k^*-2$ dependents with respect to $P'$.
Moreover, since $T$ has size at most $k^*-1$, each of the new cut vertices in $C(P')$ has at most $k^*-2$ dependents.
Hence we have decreased the number of cut vertices with $k^*-1$ dependents by at least $1$.
This gives the desired contradiction.
The same argument shows that at least $r$ vertices of $G$ (excluding $v$) belong to some block in $R_v$.

Consider two agents who have the same binary utility function with value $1$ for $v$, its $k^*-1$ dependents, $r$ arbitrary vertices of $G$ (besides $v$) belonging to some block in $L_v$, and $r$ arbitrary vertices of $G$ (besides $v$) belonging to some block in $R_v$, and value $0$ for the remaining vertices.
The total value of an agent is $2r+k^*$.
In any connected allocation, one of the agents does not receive $v$.
This agent receives value at most $r$, while the remaining goods are worth at least $r+k^*$.
It follows that the allocation cannot be EF$(k^*-1)$.

\begin{figure}[!ht]
\centering
\subfloat{
\begin{tikzpicture}[scale=0.85]
\draw[line cap=round, yellow!40, line width=8.5mm] (0,4)--(6,4);
\draw[line cap=round, blue!20, line width=5mm] (0,4)--(1.5,4);
\draw[line cap=round, green!30, line width=5mm] (3,2.5)--(1.94,1.44);
\draw[line cap=round, orange!40, line width=5mm] (4.5,4)--(6,4);
\draw[line cap=round, orange!40, line width=5mm] (4.5,2.5)--(4.5,4);
\draw (0,4) -- (6,4);
\draw (3,2.5) -- (3,4);
\draw (4.5,2.5) -- (4.5,4);
\draw (3,2.5) -- (1.94,1.44);
\draw (3,2.5) -- (4.06,1.44);
\draw [fill=blue] (0,4) circle [radius = 0.1];
\draw [fill=white] (1.5,4) circle [radius = 0.1];
\draw [fill=blue] (3,4) circle [radius = 0.1];
\draw [fill=white] (4.5,4) circle [radius = 0.1];
\draw [fill=blue] (6,4) circle [radius = 0.1];
\draw [fill=white] (3,2.5) circle [radius = 0.1];
\draw [fill=blue] (1.94,1.44) circle [radius = 0.1];
\draw [fill=blue] (4.06,1.44) circle [radius = 0.1];
\draw [fill=blue] (4.5,2.5) circle [radius = 0.1];
\node at (6.6,4) {$P$};
\node at (0.75,3.5) {$L_B$};
\node at (2.15,2.35) {$P'$};
\node at (5.1,3.1) {$R_B$};
\node at (3.25,2.5) {$v$};
\node at (3.2,3.75) {$B$};
\node at (1.5,3.75) {$v'$};
\end{tikzpicture}
}
\vspace{5mm}
\subfloat{
\begin{tikzpicture}[scale=0.85]
\draw[line cap=round, yellow!40, line width=8.5mm] (0,4)--(6,4);
\draw[line cap=round, green!30, line width=5mm] (0,4)--(1.5,4);
\draw[line cap=round, blue!20, line width=5mm] (3,2.5)--(3,1);
\draw[line cap=round, orange!40, line width=5mm] (4.5,4)--(6,4);
\draw[line cap=round, orange!40, line width=5mm] (4.5,2.5)--(4.5,4);
\draw (0,4) -- (6,4);
\draw (1.5,2.5) -- (1.5,4);
\draw (3,1) -- (3,4);
\draw (4.5,2.5) -- (4.5,4);
\draw [fill=blue] (0,4) circle [radius = 0.1];
\draw [fill=white] (1.5,4) circle [radius = 0.1];
\draw [fill=blue] (3,4) circle [radius = 0.1];
\draw [fill=white] (4.5,4) circle [radius = 0.1];
\draw [fill=blue] (6,4) circle [radius = 0.1];
\draw [fill=blue] (4.5,2.5) circle [radius = 0.1];
\draw [fill=white] (3,2.5) circle [radius = 0.1];
\draw [fill=blue] (1.5,2.5) circle [radius = 0.1];
\draw [fill=blue] (3,1) circle [radius = 0.1];
\node at (6.6,4) {$P''$};
\node at (0.75,3.5) {$P'$};
\node at (2.45,1.8) {$L_B$};
\node at (5.1,3.1) {$R_B$};
\node at (1.67,3.75) {$v$};
\node at (3.2,3.75) {$B$};
\node at (3.3,2.55) {$v'$};
\end{tikzpicture}
}
\caption{An example of a switch operation on the tree $B(G)$ in Case~2 of the proof of Theorem~\ref{thm:EF-two-EFk}.
The left and right figures are the trees before and after the operation, respectively.}
\label{fig:EFk-case2}
\end{figure}

\item Case 2: $v$ is not in the path $P$ but belongs to some block $B$ in $P$.
Let $L_B$ and $R_B$ be the subtree of the tree $B(G)$ rooted at $B$ starting with each of the two cut vertices adjacent to $B$ on the path $P$, respectively.
We claim that at least $k^*$ vertices of $G$ belong to some block in $L_B$.
Assume for contradiction that there are at most $k^*-1$ such vertices.
Let $v'$ be the cut vertex in $L_B$ adjacent to $B$.
In $B(G)$, we switch $L_B$ with $v$ and its dependents, and choose an arbitrary path $P'$ that starts with $v$ and contains at least one of its dependents as well as a leaf of $B(G)$ to be on the main path $P$  (see Figure~\ref{fig:EFk-case2}). Let $P''$ denote the new maximal path.
Since $L_B$ contains at most $k^*-1$ vertices (which include $v'$), $v'$ now has at most $k^*-2$ dependents with respect to $P''$.
Moreover, the subtree that replaced $L_B$ has at most $k^*$ vertices.
Among these vertices, $v$ and at least one other vertex belong to $P'$, which is now on the new path $P''$, so any new cut vertex has at most $k^*-2$ dependents.
Hence we have decreased the number of cut vertices with $k^*-1$ dependents by at least $1$.
This gives the desired contradiction.
The same argument shows that at least $k^*$ vertices of $G$ belong to some block in $R_B$.

Consider two agents who have the same binary utility function with value $1$ for $v$, its $k^*-1$ dependents, $k^*$ arbitrary vertices of $G$ belonging to some block in $L_B$, and $k^*$ arbitrary vertices of $G$ belonging to some block in $R_B$, and value $0$ for the remaining vertices.
The total value of an agent is $3k^*$.
In any connected allocation, one of the agents receives a bundle whose vertices of value $1$ are contained in $L_B$, $R_B$, or the set with $v$ and its dependents.
This agent receives value at most $k^*$, while the remaining goods are worth at least $2k^*$.
It follows that the allocation cannot be EF$(k^*-1)$.
\end{itemize}
Hence, in both cases there exist identical binary utility functions for which the graph does not admit an EF$(k^*-1)$ allocation, as claimed.
\end{proof}

Theorem~\ref{thm:EF-two-EFk} allows us to determine in polynomial time the optimal $k$ such that a given graph always admits an EF$k$ allocation, as well as to compute such an allocation.
To do so, we compute the block decomposition $B(G)$ of the graph---this can be done in linear time \citep{HopcroftTa73}.
We then determine the value of $k^*$ in the proof of the theorem, which we have shown to be equal to the optimal value of $k$; this can be done by testing all pairs of vertices as endpoints of the path $P$.
Finally, we compute a bipolar ordering of the vertices belonging to $P$---again, this takes linear time \citep{EvenTa76}---and apply the EF1 algorithm of \citet{BiloCaFl19} on the merged vertices.

Theorem~\ref{thm:EF-two-EFk} also yields a short proof that every graph admits an EF$(m-2)$ allocation. Moreover, we show that the bound $m-2$ is tight for stars.

\begin{proposition}
\label{prop:EF-two-m-2}
Let $n=2$, and let $G$ be any graph with at least three vertices. There exists a connected EF$(m-2)$ allocation to the two agents.
\end{proposition}

\begin{proof}
Since the graph contains at least three vertices, it has a path of length $2$; let the three vertices on this path be $v_1,v_2,v_3$.
Let $V_1=\{v_1\}, V_2=\{v_2\}, V_3=\{v_3\}$.
We add the remaining vertices to these sets arbitrarily so that each set remains connected.
Clearly, each set contains at most $m-2$ vertices.
Theorem~\ref{thm:EF-two-EFk} then implies that an EF$(m-2)$ allocation exists.
\end{proof}

\begin{proposition}
\label{prop:EF-star}
Let $n=2$, and let $G$ be a star with at least two edges. There exist identical binary utility functions of the two agents such that a connected EF$(m-3)$ allocation does not exist.
\end{proposition}

\begin{proof}
Consider two agents who have value $1$ for every good.
In any connected allocation, one of the agents receives at most one good, while the other agent receives at least $m-1$ goods.
Hence the allocation cannot be EF$(m-3)$.
\end{proof}

Next, we consider a stronger fairness notion, EFX. It is known that for two agents with arbitrary monotonic utilities, an EFX allocation always exists \citep{PlautRo20}. We show that if we consider connected allocations, the statement remains true only if the graph is complete.

\begin{theorem}
\label{thm:EF-two-EFX}
Let $n=2$, and let $G$ be a non-complete graph. There exist identical additive utility functions of the two agents such that no connected allocation is EFX.
\end{theorem}

\begin{proof}
Pick an arbitrary missing edge of $G$, and let $\epsilon > 0$ be a sufficiently small constant.
Suppose that the two agents have value $2$ for each of the two vertices with a missing edge (call them $v_1$ and $v_2$), and value $3,\epsilon,\epsilon,\dots,\epsilon$ for the remaining vertices (call the first vertex $v_3$).
Assume for contradiction that there exists a connected EFX allocation.
In this allocation, neither of the agents can receive $v_3$ together with one (or both) of $v_1,v_2$.
So one of the agents must receive $v_1$ and $v_2$, while the other agent receives $v_3$.
If the first agent also receives one of the remaining vertices, the allocation cannot be EFX.
So the second agent receives all of the remaining vertices.
However, the resulting allocation is not connected, a contradiction.
\end{proof}

\subsection{Three Agents}
\label{sec:envyfree-three}

We now address the case of three agents. \citet{BiloCaFl19} showed that in this case, an EF1 allocation is guaranteed to exist if the graph contains a Hamiltonian path\footnote{Clearly, it suffices to prove the claim when the graph is a path.} or if it is a star with three edges.
We extend this result by characterizing all trees and complete bipartite graphs that always admit an EF1 allocation.
Recall that we allow agents to have arbitrary monotonic utilities in this section.

\begin{theorem}
\label{thm:EF-three-tree}
Let $G$ be a tree. Then $G$ guarantees EF1 for three agents if and only if $G$ is either a path, or a star with three edges.
\end{theorem}

\begin{proof}
The `if' direction was already shown by \citet{BiloCaFl19}; we establish the `only if' direction.
Assume that $G$ is neither a path, nor a star with three edges.
Suppose first that there is a vertex $v$ with degree at least $4$.
Consider three agents who have identical utilities with value $1$ on $v$ and four of its neighbors, and $0$ on all other vertices.
In any connected allocation, an agent who does not get $v$ receives value at most $1$, while the bundle of the agent who gets $v$ has value at least $3$ to her.
Hence the allocation is not EF1.

\begin{figure}[!ht]
\centering
\begin{tikzpicture}[scale=0.85]
\draw (2,2) -- (8,2);
\draw (2,2) -- (0.59,0.59);
\draw (2,2) -- (0.59,3.41);
\draw [fill] (2,2) circle [radius = 0.1];
\draw [fill] (4,2) circle [radius = 0.1];
\draw [fill] (6,2) circle [radius = 0.1];
\draw [fill] (8,2) circle [radius = 0.1];
\draw [fill] (0.59,0.59) circle [radius = 0.1];
\draw [fill] (0.59,3.41) circle [radius = 0.1];
\node at (2,2.3) {$2$};
\node at (2,1.7) {$v$};
\node at (0.59,0.89) {$2$};
\node at (0.59,0.29) {$v_4$};
\node at (0.59,3.71) {$2$};
\node at (0.59,3.11) {$v_3$};
\node at (4,2.3) {$3$};
\node at (4,1.7) {$v_1$};
\node at (6,2.3) {$4$};
\node at (6,1.7) {$v_2$};
\node at (8,2.3) {$0$};
\end{tikzpicture}
\caption{Example of an instance in the proof of Theorem~\ref{thm:EF-three-tree}.}
\label{fig:EF-tree}
\end{figure}

Suppose now that every vertex has degree at most $3$. Since $G$ is not a path, there is a vertex $v$ with degree $3$.
Moreover, since $G$ is not a star, one of the branches from $v$ contains at least two vertices, say a branch starting with a neighbor $v_1$ of $v$ followed by another vertex $v_2$.
Let $v_3,v_4$ be the two other vertices adjacent to $v$.
Consider three agents who have identical utilities with value $2$ for $v,v_3,v_4$, value $3$ for $v_1$, value $4$ for $v_2$, and value $0$ for all other vertices (see Figure~\ref{fig:EF-tree}).
Consider any connected allocation; in what follows, we will only be concerned with goods of non-zero value.
First, assume that one of the agents receives either only $v_3$ or only $v_4$, and obtains value at most $2$.
If another agent receives at least three goods, the allocation is clearly not EF1.
So each of the other two agents receives exactly two goods, which means one of them receives $v_1$ and $v_2$.
This bundle is worth $3$ to the first agent even after removing the most valuable good, so the allocation cannot be EF1.
Hence one of the agents receives $v_3$, $v_4$, and $v$.
But then the agent who does not receive $v_2$ will envy this agent even after removing one good.
\end{proof}

Next, we consider complete bipartite graphs.
Denote by $K_{a,b}$ the complete bipartite graph with $a$ vertices on the left (call this set of vertices $L$) and $b$ vertices on the right (call this set of vertices $R$).
We start by showing that if $a,b\geq 3$, there always exists a connected EF1 allocation.
In fact, we present a generalization that holds for any number of agents.

\begin{proposition}
\label{prop:EF-three-3-3}
Let $n\geq 2$, and let $G$ be a complete bipartite graph $K_{a,b}$ with $a,b\geq n$. Then $G$ guarantees EF1 for $n$ agents.
\end{proposition}

\citet{BiloCaFl19} posed the question of finding an infinite class of graphs without a Hamiltonian path that guarantee EF1 for three or more agents.
Since $K_{a,b}$ does not contain a Hamiltonian path whenever $|a-b|\ge 2$, Proposition~\ref{prop:EF-three-3-3} answers their question for every $n \ge 3$.

\begin{proof}[Proof of Proposition~\ref{prop:EF-three-3-3}]
We enhance the \emph{envy cycle elimination algorithm} of \citet{LiptonMaMo04}, which computes an EF1 allocation for any number of agents.
The algorithm works by allocating one good at a time in arbitrary
order---we will exploit this freedom in choosing the order.
It also maintains an \emph{envy graph}, which has the agents as its vertices, and a directed edge $i\rightarrow j$ if agent $i$ envies agent $j$ with respect to the current (partial) allocation.
At each step, the next good is allocated to an agent with no incoming
edge (i.e., the agent is unenvied), and any cycle that arises as a result is eliminated by giving $j$’s bundle to $i$ for each edge $i\rightarrow j$ in the cycle.
This allows the algorithm to maintain the invariant that the envy graph is cycle-free, and so there exists an agent with no incoming edge before each allocation of a good.

We apply the envy cycle elimination algorithm by choosing a careful order of the goods to allocate.
Since $a\geq n$ and every agent is unenvied at the beginning, we can first pick $n$ goods from $L$ and allocate one of them to each agent.
After this point, we may not simply pick an arbitrary agent to be allocated the next good, since that agent may be envied.
We take an unenvied agent, i.e., an agent with no incoming edge in the envy graph, and consider two cases.

\underline{Case 1}: The agent has already received a good from $R$.
In this case, we allocate to her a good from $L$ if one still remains; otherwise, we give her a good from $R$.

\underline{Case 2}: The agent has not received a good from $R$.
In this case, we allocate to her a good from $R$ if one still remains; otherwise, we give her a good from $L$.

The pseudocode of the algorithm is presented as Algorithm~\ref{alg:envy-cycle}.

\begin{algorithm}
\caption{Enhanced Envy Cycle Elimination Algorithm}\label{alg:envy-cycle}
\begin{algorithmic}[1]
\Procedure{EnhancedEnvyCycleElimination$(N,L,R,u_1,\dots,u_n)$}{}
\State $M_1,\dots,M_n\leftarrow\emptyset$
\State $r_1,\dots,r_n\leftarrow \text{false}$ // $r_i$ indicates whether agent~$i$ has received a good from $R$
\For{$i \in N$}
\State Move an arbitrary good from $L$ to $M_i$.
\EndFor
\While{$L\cup R\neq\emptyset$}
\State Eliminate cycles in the envy graph.
\State $i\leftarrow \text{any agent with no incoming edge in the envy graph}$
\If{$r_i = \text{true}$}
\If{$L\neq\emptyset$}
\State Move an arbitrary good from $L$ to $M_i$.
\Else
\State Move an arbitrary good from $R$ to $M_i$.
\EndIf
\Else
\If{$R\neq\emptyset$}
\State Move an arbitrary good from $R$ to $M_i$.
\State $r_i\leftarrow\text{true}$
\Else
\State Move an arbitrary good from $L$ to $M_i$.
\EndIf
\EndIf
\EndWhile
\State \Return $(M_1,\dots,M_n)$
\EndProcedure
\end{algorithmic}
\end{algorithm}

The resulting allocation is EF1 \citep{LiptonMaMo04}; we now show that it is connected.
Every agent receives a good from $L$ in the first phase of the algorithm.
Note that if an agent receives at least one good from both $L$ and $R$, her bundle is guaranteed to be connected.
So it suffices to show that an agent will never receive more than one good from $L$ without receiving a good from $R$.
By construction, an agent who already has a good from $R$ will take goods from $L$ unless $L$ is already empty.
Since $b\geq n$, this means that as long as some agent has not received a good from $R$ and the algorithm has not terminated, there is at least one good from $R$ left.
This establishes the desired claim.
\end{proof}

Note that since the envy cycle elimination algorithm runs in time polynomial in the number of agents and goods \citep{LiptonMaMo04}, the proof of Proposition~\ref{prop:EF-three-3-3} also yields a polynomial-time algorithm that computes a connected EF1 allocation for any number of agents.
If the agents have additive utilities, we can also obtain an EF1 allocation via a ``double round-robin algorithm''---the details can be found in Appendix~\ref{app:double-round-robin}.

With Proposition~\ref{prop:EF-three-3-3} in hand, we now proceed with the characterization for complete bipartite graphs.

\begin{theorem}
\label{thm:EF-three-bipartite}
Let $a,b$ be positive integers with $a\leq b$.
The graph $K_{a,b}$ guarantees EF1 for three agents if and only if one of the following holds:
\begin{enumerate}
    \item $a=1$ and $b\leq 3$;
    \item $a=2$ and $b\leq 3$;
    \item $a,b\geq 3$.
\end{enumerate}
\end{theorem}

\begin{figure}[!ht]
\centering
\begin{tikzpicture}[scale=0.85]
\draw [fill] (2,4) circle [radius = 0.1];
\draw [fill] (2,6) circle [radius = 0.1];
\draw [fill] (6,3) circle [radius = 0.1];
\draw [fill] (6,4) circle [radius = 0.1];
\draw [fill] (6,5) circle [radius = 0.1];
\draw [fill] (6,6) circle [radius = 0.1];
\draw [fill] (6,7) circle [radius = 0.1];
\draw (2,4) -- (6,3);
\draw (2,4) -- (6,4);
\draw (2,4) -- (6,5);
\draw (2,4) -- (6,6);
\draw (2,4) -- (6,7);
\draw (2,6) -- (6,3);
\draw (2,6) -- (6,4);
\draw (2,6) -- (6,5);
\draw (2,6) -- (6,6);
\draw (2,6) -- (6,7);
\node at (1.6,6) {$2$};
\node at (2,5.7) {$v_1$};
\node at (1.6,4) {$2$};
\node at (2,3.7) {$v_2$};
\node at (6.4,7) {$1$};
\node at (6,6.7) {$v_3$};
\node at (6.4,6) {$1$};
\node at (6,5.7) {$v_4$};
\node at (6.4,5) {$1$};
\node at (6,4.7) {$v_5$};
\node at (6.4,4) {$1$};
\node at (6,3.7) {$v_6$};
\node at (6.4,3) {$0$};
\end{tikzpicture}
\caption{Example of an instance in the proof of Theorem~\ref{thm:EF-three-bipartite}.}
\label{fig:EF-bipartite}
\end{figure}

\begin{proof}
The case $a=1$ is covered by Theorem~\ref{thm:EF-three-tree} and the case $a\geq 3$ by Proposition~\ref{prop:EF-three-3-3}, so assume that $a=2$.
If $b\leq 3$, then $G$ contains a Hamiltonian path, so the existence of an EF1 allocation follows from the result of \citet{BiloCaFl19}.
Else, let $b\geq 4$. Consider three agents who have identical utilities with value $2$ on each of the two vertices $v_1,v_2\in L$, value $1$ on four of the vertices $v_3,v_4,v_5,v_6\in R$, and value $0$ for the remaining vertices (see Figure~\ref{fig:EF-bipartite}).
Consider any connected allocation.
If $v_1$ and $v_2$ are allocated to the same agent, this agent must also receive at least one of the vertices from $R$, and the allocation is not EF1.
Else, one agent receives $v_1$ and another agent receives $v_2$.
Now, the third agent can get at most one vertex from $R$ and therefore receives value at most $1$.
This means that one of the first two agents receives one of $v_1$ and $v_2$ along with at least two of $v_3,v_4,v_5,v_6$.
This agent is envied by the third agent even after we remove a good.
It follows that the allocation cannot be EF1.
\end{proof}

\section{Conclusion and Future Work}

In this paper, we study the fair allocation of indivisible goods under connectivity constraints and provide an extensive set of results on the guarantees that can be achieved via maximin share fairness and relaxations of envy-freeness for various classes of graphs.
For maximin share fairness, we establish a link between the graph-specific maximin share and the well-studied maximin share through our price of connectivity (PoC) notion. We present a number of bounds on the PoC, several of which are tight, and leave a tempting conjecture that would settle the two-agent case if it holds.
On the envy-freeness front, we classify all connected graphs based on the strongest relaxation with guaranteed existence in the case of two agents---thereby also quantifying the price that we have to pay with respect to fairness for each graph---and characterize the set of trees and complete bipartite graphs that always admit an EF1 allocation for three agents.
Extending our results beyond three agents is a challenging problem: even when the graph is a path, the only known proof of EF1 existence for four agents employs arguments based on Sperner's lemma, and the corresponding question remains open when there are at least five agents \citep{BiloCaFl19}.

Our results on envy-freeness relaxations hold for agents with arbitrary monotonic utilities.
On the other hand, as is the case in most of the literature, our results on maximin share fairness rely on the assumption that the agents' utility functions are additive.
Maximin share fairness beyond additive utilities has been studied by \citet{BarmanKr20} and \citet{GhodsiHaSe18}; for example, they showed that a constant approximation of the maximin share can be achieved for any number of agents with submodular utilities when the graph is complete.
Since complementarity and substitutability are common in practice, it would be interesting to see how the graph-based approximations that we obtain in this paper change as we enlarge the class of utility functions considered.
Indeed, as \citet{PlautRo20} noted, there is a rich landscape of problems to explore in fair division with different classes of utility functions, and the graphical setting is likely to be no exception.

Finally, while our results in this work provide fairness guarantees that hold regardless of the agents' utilities, better guarantees can be obtained in many instances if we take the utilities into account.
For example, even though an envy-free allocation does not always exist, it is known that such an allocation exists most of the time when utilities are drawn at random \citep{DickersonGoKa14,ManurangsiSu19,ManurangsiSu21}.
On a complete graph, deciding the existence of an envy-free allocation is NP-hard even for two agents with identical utilities \citep{LiptonMaMo04}.
By contrast, this problem can be solved efficiently on a tree or a cycle for any constant number of agents, since we can simply go through all of the (polynomially many) connected allocations; yet, the problem again becomes NP-hard even on a path if the number of agents is non-constant \citep{BouveretCeEl17}.
Similar computational questions can be asked for other combinations of graphs and fairness notions without guaranteed existence, and we believe that these questions constitute an important direction that deserves to be pursued in future work.

\section*{Acknowledgments}

This work was partially supported
by the Ministry of Education, Singapore, under its Academic Research Fund Tier~1 (RG23/20),
by the KAKENHI Grant-in-Aid for JSPS Fellows number 18J00997,
by the European Research Council (ERC) under grant number 639945 (ACCORD),
by an NUS Start-up Grant, and by JST, ACT-X.

\bibliographystyle{plainnat}
\bibliography{arxiv}

\appendix

\section{Algorithm for Theorem~\ref{thm:mms-two-connectivity-2}} \label{app:algo-connectivity-2}
In this section, we give a polynomial-time algorithm for computing an allocation that gives both agents at least $3/4$ of their MMS when the graph is biconnected.
As in the remark following Theorem~\ref{thm:mms-two-connectivity-1}, it suffices to compute a bipartition such that the first agent has value at least $3/4$ of her MMS for both parts; the second agent can then choose the part that she prefers.

To compute such a bipartition, we iterate over all pairs of goods $g_1,g_2$.
For each pair, we construct a bipolar ordering that begins with $g_1$ and ends with $g_2$; this is possible as explained in the proof of Theorem~\ref{thm:mms-two-connectivity-2}.
We then consider taking every possible prefix of the ordering as one part of the bipartition, and return the bipartition with the highest minimum between the two parts across all pairs $g_1,g_2$.
The pseudocode of the algorithm is given as Algorithm~\ref{alg:biconnected}.

\begin{algorithm}
\caption{Approximate MMS Algorithm for Biconnected Graphs}\label{alg:biconnected}
\begin{algorithmic}[1]
\Procedure{ApproximateMMS$(M,G,u)$}{}
\State current-best $\leftarrow$ 0
\State $M_1 \leftarrow \emptyset$
\State $M_2 \leftarrow M$
\For{$(g_1,g_2)\in M\times M$ with $g_1\neq g_2$}
\State Construct a bipolar ordering $\sigma$ of $G$ starting with $g_1$ and ending with $g_2$.
\For{$g\in M$}
\State $M_1' \leftarrow \{$all goods before $g$ in $\sigma\}$
\State $M_2' \leftarrow \{g$ and all goods after $g$ in $\sigma\}$
\If{$\min\{u(M_1'),u(M_2')\} >$ current-best}
\State current-best $\leftarrow \min\{u(M_1'),u(M_2')\}$
\State $M_1 \leftarrow M_1'$
\State $M_2 \leftarrow M_2'$
\EndIf
\EndFor
\EndFor
\State \Return $(M_1,M_2)$
\EndProcedure
\end{algorithmic}
\end{algorithm}

Since constructing a bipolar ordering with a specific cycle as the first ear can be done in linear time \citep{Schmidt13}, Algorithm~\ref{alg:biconnected} runs in polynomial time.
We now establish the correctness of the algorithm.
Assume without loss of generality that $\text{MMS}(u,2)=1/2$, so there exists a bipartition $(M_1,M_2)$ of $M$ such that $u(M_1)=1/2\leq u(M_2)$.\footnote{Note that we are not assuming $u(M)=1$ as in the proof of Theorem~\ref{thm:mms-two-connectivity-2}.}
Consider a modified utility function $u'$ where we start with $u$ and arbitrarily decrease the values of some goods in $M_2$ so that $u'(M_2) = 1/2$.
In the new instance, the proof of Theorem~\ref{thm:mms-two-connectivity-2} implies that there exists a connected bipartition for which both parts have value at least $3/8$, and this bipartition corresponds to one of the bipartitions examined by Algorithm~\ref{alg:biconnected}.
Since the values in the original instance with utility function $u$ can only be higher than in the new instance with utility function $u'$, in the original instance both parts of this bipartition also have value at least $3/8$.
It follows that both parts of the bipartition returned by Algorithm~\ref{alg:biconnected} have value at least $3/8$, which is $3/4$ of the MMS.

\section{Double Round-Robin Algorithm} \label{app:double-round-robin}
In this section, we provide a simple polynomial-time algorithm for computing a connected EF1 allocation among $n$ agents with \emph{additive} utilities, when the graph is a complete bipartite graph $K_{a,b}$ with $a,b\geq n$.
Let $L$ and $R$ denote the set of vertices on the left and right side of the graph, respectively.
The algorithm proceeds by running the classical round-robin algorithm twice, once on $L$ and once on $R$, with opposite orderings of the agents.
The pseudocode is shown as Algorithm~\ref{alg:roundrobin}.

\begin{algorithm}
\caption{Double Round-Robin Algorithm}\label{alg:roundrobin}
\begin{algorithmic}[1]
\Procedure{DoubleRoundRobin$(N,L,R,u_1,\dots,u_n)$}{}
\State $M_1,\dots,M_n\leftarrow \emptyset$
\State $i=1$
\While{$L\neq\emptyset$}
\State $j\leftarrow\text{highest-valued good in }L\text{ according to }u_i$
\State{Move $j$ from $L$ to $M_i$.}
\State{$i\leftarrow i+1$}
\If{$i=n+1$}
\State $i=1$
\EndIf
\EndWhile
\State $i=n$
\While{$R\neq\emptyset$}
\State $j\leftarrow\text{highest-valued good in }R\text{ according to }u_i$
\State{Move $j$ from $R$ to $M_i$.}
\State{$i\leftarrow i-1$}
\If{$i=0$}
\State $i=n$
\EndIf
\EndWhile
\State \Return $(M_1,\dots,M_n)$
\EndProcedure
\end{algorithmic}
\end{algorithm}

Since $a,b\geq n$, every agent receives at least one good from each of $L$ and $R$, so the resulting allocation is connected.
We claim that it is EF1.
To see this, consider two agents $i,i'$ with $i<i'$.
When allocating each of the sets $L$ and $R$, we consider a \emph{round} to begin when $i$ picks a good, and end just before the next time $i$ picks a good (or when the set runs out of goods).
During the allocation of $L$, in each round $i$ picks before $i'$. Since the utilities are additive, $i$ does not envy $i'$ with respect to the goods in $L$.
Similarly, $i$ does not envy $i'$ in each round during the allocation of $R$.
The only possible source of envy is before the first round starts, when $i'$ picks her first good.
However, this means that the envy can be eliminated if we remove this good from the bundle of $i'$.
Hence $i$ does not envy $i'$ up to one good in total; an analogous argument shows that $i'$ also does not envy $i$ up to one good.
Since $i$ and $i'$ are arbitrary, the allocation is EF1.

\end{document}